\documentclass[11pt]{article}

\usepackage[utf8]{inputenc}
\usepackage{amsmath,amssymb,amsthm}
\usepackage{geometry}
\usepackage{graphicx}
\usepackage{thmtools, thm-restate}
\usepackage[colorlinks,citecolor=blue,urlcolor=blue,linkcolor=blue,bookmarks=false]{hyperref}
\usepackage{listings}
\usepackage[framemethod=tikz]{mdframed}
\usepackage[sort]{natbib}
\usepackage{float}
\usepackage{makecell}
\usepackage{cancel}
\usepackage[shortlabels]{enumitem}
\usepackage{bm}
\usepackage{caption, subcaption}
\usepackage{mdframed}
\usepackage{booktabs}

\allowdisplaybreaks

\definecolor{light-gray}{gray}{0.985}
\mdfsetup{%
	linecolor=white,
	backgroundcolor=gray!20,
}
\definecolor{light-purple}{rgb}{0.85, 0.75, 1.0}

\newtheorem{theorem}{Theorem}
\newtheorem{lemma}{Lemma}
\newtheorem{corollary}{Corollary}
\newtheorem{proposition}{Proposition}

\newtheorem{algorithm}{Algorithm}

\theoremstyle{definition}
\newtheorem{definition}{Definition}

\newtheorem{smooth-cross-example}{Cross-world S-flip intervention}

\theoremstyle{remark}
\newtheorem{assumption}{Assumption}
\newtheorem{remark}{Remark}

\def\inprob{\stackrel{p}{\rightarrow}}
\def\indist{\rightsquigarrow}
\newcommand\ind{\protect\mathpalette{\protect\independenT}{\perp}}
\def\independenT#1#2{\mathrel{\rlap{$#1#2$}\mkern4mu{#1#2}}}

\DeclareSymbolFont{bbold}{U}{bbold}{m}{n}
\DeclareSymbolFontAlphabet{\mathbbold}{bbold}
\newcommand{\one}{\mathbbold{1}}
\def\cov{\text{cov}}
\def\bbE{\mathbb{E}}
\def\bbP{\mathbb{P}}
\def\bbR{\mathbb{R}}
\def\bbV{\mathbb{V}}

\usepackage{stackengine}  
\stackMath                

\title{Longitudinal weighted and trimmed treatment effects \\ with flip interventions} 
\date{}
\author{Alec McClean\textsuperscript{1}\footnote{Corresponding author: \texttt{hadera01@nyu.edu}}\ , Alexander W. Levis\textsuperscript{2}, Nicholas Williams\textsuperscript{3}, and Iv\'{a}n D\'{i}az\textsuperscript{1} \\ \\
\small \textsuperscript{1}\emph{Division of Biostatistics, New York University Grossman School of Medicine} \\
\small \textsuperscript{2}\emph{Department of Statistics and Data Science, Carnegie Mellon University} \\
\small \textsuperscript{3}\emph{Department of Epidemiology, Columbia University Mailman School of Public Health}}

\begin{document}
\maketitle

\begin{abstract}
    Weighting and trimming are popular methods for addressing positivity violations in causal inference. While well-studied with single-timepoint data, standard methods do not easily generalize to address non-baseline positivity violations in longitudinal data, and remain vulnerable to such violations. In this paper, we extend weighting and trimming to longitudinal data via stochastic ``flip'' interventions, which maintain the treatment status of subjects who would have received the target treatment, and flip others' treatment to the target with probability equal to their weight (e.g., overlap weight, trimming indicator). We first show, in single-timepoint data, that flip interventions yield a large class of weighted average treatment effects, ascribing a novel policy interpretation to these popular weighted estimands. With longitudinal data, we then show that flip interventions provide interpretable weighting or trimming on non-baseline covariates and, crucially, yield effects that are identifiable under arbitrary positivity violations. Moreover, we demonstrate that flip interventions are policy-relevant since they could be implemented in practice. By contrast, we show that alternative approaches for weighting on non-baseline covariates fail to achieve this property.  We derive flexible and efficient estimators based on efficient influence functions when the weight is a smooth function of the propensity score. Namely, we construct multiply robust-style and sequentially doubly robust-style estimators that achieve root-n consistency and asymptotic normality under nonparametric conditions. Finally, we demonstrate our methods through an analysis of the effect of union membership on earnings.
\end{abstract}


\section{Introduction}

There is a large and fast-growing literature in causal inference for estimating treatment effects from observational data under violations of the positivity assumption. For binary treatments, positivity requires subjects in every covariate stratum to have non-zero probability of receiving both treatment and control \citep{hernan2020whatif}. In longitudinal settings, positivity requires non-zero probability for every treatment regime under consideration, which can become increasingly difficult to satisfy with more timepoints as the number of treatment regimes increases exponentially. Exact positivity violations (zero probability of certain regimes) render causal effects unidentifiable, while ``practical'' violations (near-zero probability of certain regimes) can prohibitively inflate variance estimates. Both hamper meaningful scientific conclusions from data \citep{kang2007demystifying, moore2012causal, petersen2012diagnosing}.

\bigskip

Several approaches have been developed to define effects robust to positivity violations. We focus on weighting/trimming, which we relate to our novel flip interventions. In single-timepoint data, weighted average treatment effects (WATEs) take the form
\begin{equation} \label{eq:wate-intro}
    \bbE \left[ \frac{\bbE \{ Y(1) - Y(0) \mid X \} f(X)}{\bbE \{ f(X) \}} \right]
\end{equation}
where $Y(1)$ and $Y(0)$ are potential outcomes, $X$ are covariates, and $f(X)$ is a non-negative weight, which is typically based on the propensity score $\pi(X) := \bbP(A=1 \mid X)$. Certain WATEs are estimable under positivity violations, namely those that assign zero weight to covariate values for which positivity fails. Examples include overlap weights \citep{li2018balancing}, entropy weights \citep{hainmueller2012entropy}, and trimmed or smooth trimmed effects \citep{crump2009dealing, yang2018asymptotic}.  WATEs can also be defined implicitly post hoc by estimating weights $f(X)$ to directly maximize covariate balance between treatment and control groups and then estimating the identified analogue of the WATE in \eqref{eq:wate-intro} \citep{cohn2023balancing}. 

\bigskip

Weighting and trimming are well-studied for single-timepoint data, with extensions beyond binary treatment (e.g., \citet{li2019propensity, branson2023causal}), but their extension to longitudinal settings remains an important unresolved challenge. Current approaches typically trim using only baseline covariates \citep{jensen2024identification, petersen2012diagnosing}, leaving them vulnerable to positivity violations from non-baseline covariates. While several longitudinal weighting methods exist, they assume that positivity holds. Rather than addressing violations of positivity, they focus on minimizing estimation error \citep{imai2015robust, zeng2023propensity, viviano2021dynamic}.  

\bigskip

Meanwhile, an alternative approach has developed which focuses on constructing interventions which adapt to positivity violations at each timepoint and yield identifiable effects.  Pioneered by \citet{robins2004effects} and \citet{stock1989nonparametric}, these interventions have evolved through various formulations including realistic interventions \citep{van2007causal}, modified treatment policies (MTPs) \citep{haneuse2013estimation, diaz2012population}, threshold interventions \citep{taubman2009intervening}, incremental propensity score interventions \citep{kennedy2019nonparametric}, and recent extensions to complex settings \citep{mcclean2024fair, schindl2024incremental, stensrud2024optimal, wen2023intervention}. In their general formulation---when the intervention can depend on covariate and treatment information, and additional auxiliary randomness (see \citet[Section 1] {diaz2023nonparametric} for further details)---we refer to these interventions as ``stochastic longitudinal MTPs'' (S-LMTPs). The natural applicability of S-LMTPs to longitudinal data has made them increasingly popular.

\subsection{Structure of the paper and our contributions}

Methods for weighting/trimming and S-LMTPs have developed largely in parallel. In Section~\ref{sec:t=1}, we demonstrate that (single-timepoint) WATEs correspond to mean differences in potential outcomes under pairs of ``flip'' interventions. Each pair contains one intervention targeting treatment and one targeting control. If a subject would have taken the target treatment, the flip intervention does not intervene; otherwise, it flips the subject to the target treatment with probability equal to the weight $f(X)$. More specifically, we show that WATEs are \emph{interventional} effects: they are the mean difference in potential outcomes under the treatment flip intervention minus the control flip intervention, \emph{per unit of additional treatment} \citep{zhou2022marginal}. These have also been referred to as ``policy-relevant treatment effects'', which represent the per capita effect of one policy versus another \citep{heckman2001policy, heckman2005structural}. This connection is of independent interest because it yields a direct policy interpretation for a large class of WATEs, including those that balance covariates directly and define the WATE post hoc.

\bigskip

Building on this insight, Section~\ref{sec:longitudinal-setup} introduces our notation for longitudinal data, while Section~\ref{sec:ltt} extends flip interventions to longitudinal data. These interventions flip subjects toward a target treatment regime using weights/flipping probabilities constructed from non-baseline covariates. Crucially, they remain identifiable even under positivity violations at all timepoints and are \emph{single-world} interventions, relying only on data observable under the intervention itself, and therefore are policy-relevant because they could be implemented in practice \citep{richardson2013single}. We define \emph{longitudinal} interventional flip effects, which capture the mean difference in potential outcomes under two flip interventions, standardized by the average change in the number of treatments per timepoint. In Section~\ref{subsec:properties}, we investigate the properties of longitudinal interventional flip effects and outline issues with alternative estimands one might consider.  Crucially, we show that direct extensions of weighting and trimming to longitudinal data do not correspond to interventions that could be implemented in practice, justifying our focus on flip interventions.

\bigskip

Section~\ref{sec:efficiency} develops efficient estimators for flip effects when the flipping probability is a smooth function of the propensity score. We present multiply robust and sequentially doubly robust estimators based on efficient influence functions that achieve parametric convergence rates under nonparametric conditions. We establish two new results: (1) tighter bounds on the bias of multiply robust estimators using the minimum of two error decompositions, and (2) the first sequentially doubly robust guarantee for a stochastic LMTP. For single-timepoint data, this result may be of independent interest because it provides a unified approach for doubly robust-style estimation for a large class of WATEs. Finally, in Section~\ref{sec:data-analysis}, we illustrate our methods through an analysis of the effect of union membership on earnings, and in Section~\ref{sec:discussion} we conclude and discuss.

\subsection{Mathematical notation}

For a function $f(Z)$, we use $\lVert f \rVert = \sqrt{\int f(z)^2 d\bbP(z)}$ to denote the $L_2(\bbP)$ norm, $\bbP(f) = \int_{\mathcal{Z}} f(z) d\bbP(z)$ to denote the average with respect to the underlying distribution $\bbP$, and $\bbP_n(f) = \frac{1}{n} \sum_{i=1}^{n} f(Z_i)$ to denote the empirical average with respect to $n$ observations.  In a standard abuse of notation, when $B$ is an event we let $\bbP(B)$ denote the probability of $B$.  We also denote expectation and variance with respect to the underlying distribution by $\bbE$ and $\bbV$, respectively.  We use $a \wedge b$ for minimum and $a \vee b$ for maximum, and $a \lesssim b$ to mean $a \leq Cb$ for some constant $C$. We use $\indist$ to denote convergence in distribution, and $\inprob$ for convergence in probability. Additionally, we use $o_\bbP(\cdot)$ to denote convergence in probability to zero, i.e., if $X_n$ is a sequence of random variables then $X_n = o_\bbP (r_n)$ implies $\left| \frac{X_n}{r_n} \right| \inprob 0$.

\section{Single-timepoint flip interventions and weighted average treatment effects} \label{sec:t=1}

We assume data $\{ (X_i, A_i, Y_i) \}_{i=1}^{n} \stackrel{iid}{\sim} \bbP$ where $X \in \bbR^d$ are covariates, $A \in \{0,1\}$ is a binary treatment, and $Y \in \bbR$ is an outcome. Moreover, we assume that the observed data derive from complete data $\{ (X_i, A_i, Y_i(0), Y_i(1) \}_{i=1}^{n} \stackrel{iid}{\sim} \bbP^c$ where $Y(a)$ is the potential outcome under treatment $a$.  We let $Y(D)$ denote the potential outcome under treatment decision $D$, where $D \in \{0,1\}$ is a random variable that can depend on the treatment $A$ and covariates $X$. Finally, we let $\pi(X) = \bbP(A=1 \mid X)$ denote the propensity score. 

\bigskip

When positivity is violated, a common estimand of interest is the weighted average treatment effect (WATE), which we denote as 
\begin{equation} \label{eq:wate}
    \bbE \left[ \frac{\bbE \{ Y(1) - Y(0) \mid X \} f(X)}{\bbE\{ f(X) \}} \right].
\end{equation}
Our main result in this section shows that a large class of WATEs can be defined via flip interventions. First, we define a pair of flip interventions using the weight function $f(X)$. Then, we define an interventional flip effect based on this pair. And finally, we establish that this interventional flip effect is exactly the WATE in \eqref{eq:wate}.
\begin{definition}[Single-timepoint flip interventions] \label{def:flip-t=1}
    Given a weight function $f: \bbR^d \to [0,1]$, which maps from the covariate information to $[0,1]$, define a pair of flip interventions, one for each $a \in \{0,1\}$, as
    \[
    D_f(0) = A\one\{V > f(X)\} \text{ and } D_f(1) = A + (1 - A)\one\{V \leq f(X)\},
    \]
    where $V \sim \text{Unif}(0,1)$ is drawn independently of $(X, A, Y(0),Y(1))$. In words:
    \begin{itemize}
        \item if the treatment, $A$, equals the target treatment $a$, the flip intervention does nothing; 
        \item otherwise, it flips the subject to the target treatment with probability $f(X)$.
    \end{itemize}
\end{definition}

\begin{remark}
    In Definition~\ref{def:flip-t=1}, we introduce an auxiliary random variable $V$. This is a standard device in the literature on stochastic interventions, and captures the idea of randomly reassigning treatment according to a Bernoulli distribution with a modified probability \citep{diaz2023nonparametric}. Thus, the flip intervention is a ``stochastic'' MTP, because it can depend on observed treatment and covariates, but also introduces additional randomness (see \citet{diaz2023nonparametric} for a review).
\end{remark} 

\noindent Next, we define an interventional flip effect based on a pair of flip interventions.

\begin{definition}[Interventional flip effect] \label{def:int-flip-effect}
    For a pair of flip interventions $\{ D_f(0), D_f(1) \}$ from Definition~\ref{def:flip-t=1}, we define the \emph{interventional} flip effect as
    \begin{equation} \label{eq:int-flip-effect}
        \psi_f = \frac{ \bbE \big[ Y\{ D_f(1) \} - Y\{ D_f(0) \} \big]}{\bbE \{ D_f(1) - D_f(0) \}}.
    \end{equation}
    This is the average effect on potential outcomes of $D_f(1)$ compared to $D_f(0)$ \emph{per unit of additional treatment}.
\end{definition}

We define $\psi_f$ as an \emph{interventional} effect, which standardizes the mean difference in potential outcomes by the mean difference in the number of treated units. Therefore, $\psi_f$ captures the notion of the treatment effect \emph{per unit treated}. This type of effect is useful for understanding the effectiveness of one policy versus another while accounting for how many subjects are affected by the change in policy.  It has been studied previously in the literature; see, e.g., \citet{zhou2022marginal} and references therein.  It has also been referred to as a \emph{per capita} effect \citep{heckman2005structural}. In this single-timepoint case, we note that this effect also has a \textit{conditional} interpretation as the average treatment effect among the flipped: $\psi_f = \mathbb{E}\{ Y(1) - Y(0) \mid V \leq f(X) \}$. However, such a straightforward conditional interpretation generally does not extend to the longitudinal setting, as the sequential nature of treatment assignment complicates the definition of a  subgroup analogous to the ``flipped'' set---see Section~\ref{subsec:properties}.

\bigskip

\noindent Finally, we establish the interventional flip effect in \eqref{eq:int-flip-effect} is exactly the WATE in \eqref{eq:wate}.

\begin{proposition}[Interventional flip effects are WATEs] \label{prop:equivalence}
    Let $\psi_f$ denote an interventional flip effect from Definition~\ref{def:int-flip-effect}. Then,
    \[
    \psi_f = \bbE \left[ \frac{\bbE \{ Y(1) - Y(0) \mid X \} f(X)}{\bbE\{ f(X) \}} \right].
    \]
\end{proposition}

Proposition~\ref{prop:equivalence} establishes that interventional flip effects are equal to weighted average treatment effects. While we leverage this result to extend weighting and trimming to longitudinal data, this equivalence may itself be of independent interest. Crucially, Proposition~\ref{prop:equivalence} demonstrates that any WATE with bounded weights corresponds to a contrast in specific interventions---a connection that, to our knowledge, has not been established previously. Table~\ref{tab:wates} illustrates several important examples, including the average treatment effect on the treated (ATT) or control (ATC), the average treatment effect on the overlap population (ATO), and the trimmed average treatment effect.\footnote{Table~\ref{tab:wates} is adapted from the helpful tutorial on weighting, here: \url{https://www2.stat.duke.edu/~fl35/OW/ICHPS2023.pdf}} Furthermore, this framework also allows researchers to retrospectively define the implicit estimand underlying any weighting estimator that directly balances covariates and constrains weights to lie within the interval $[0,1]$. For instance, if weights are constructed from a separate sample and subsequently applied to estimate a WATE in a new sample conditional on these estimated weights, the resulting estimand is equivalent to the interventional flip effect $\psi_{\widehat f}$, where the estimated weighting function $\widehat f$ is treated as fixed.

\begin{table}[ht] 
    \centering 
    \begin{tabular}{p{5cm} p{10cm}}
        \toprule
        \textbf{Estimand; $\psi_f$} & \textbf{Weight/flipping probability; \( f(X) \)} \\
        \midrule
        ATE & \( 1 \) \\   
        ATT & \( \pi(X) \) \\
        ATC & \( 1 - \pi(X) \) \\
        ATO & \( \pi(X)(1 - \pi(X)) \) \\
        Trimmed ATE & \( \one \{ \varepsilon \leq \pi(X) \leq 1 - \varepsilon \} \) for $\varepsilon \geq 0$ \\
        Smooth trimmed ATE & \( S \{ \pi(X); \varepsilon \} \) where $S(x; \varepsilon)$ approximates $\one (\varepsilon \leq x \leq 1 - \varepsilon)$ \\
        Matching-weighted ATE & \( \pi(X) \wedge 1 - \pi(X) \) \\
        Direct covariate balancing & Varies and is data-dependent; derived to directly balance covariates \\
        \bottomrule
    \end{tabular}
    \caption{Common weighted average treatment effect estimands and weights/flipping probabilities \( f(X) \).}
    \label{tab:wates}
\end{table}

\section{Setup and background for longitudinal flip interventions} \label{sec:longitudinal-setup}

As in the single-timepoint case, we assume $n$ observations drawn iid from some distribution $\bbP$ in a space of distributions $\mathcal{P}$; i.e., we observe data $\{ Z_i \}_{i=1}^{n} \stackrel{iid}{\sim} \bbP \in \mathcal{P}$. We assume each observation consists of longitudinal data over $T$ timepoints, so that
\[
Z = (X_1, A_1, X_2, A_2, \dots, X_T, A_T, Y),
\]
where $X_t \in \bbR^d$ are time-varying covariates ($X_1$ are baseline covariates), $A_t \in \{0,1\}$ is a time-varying binary treatment, and $Y \in \bbR$ is the ultimate outcome of interest. For a time-varying random variable $O_t$, let $\overline O_t = (O_1, \dots, O_t)$ denote its history up to time $t$ and $\underline O_t = (O_t, \dots, O_T)$ denote its future from time $t$. Let $H_t = (\overline X_t, \overline A_{t-1})$ denote covariate and treatment history up until treatment in timepoint $t$. 

\bigskip

We formalize the definition of causal effects using a nonparametric structural equation model (NPSEM) \citep{pearl2009causality}. We assume the existence of deterministic functions $\{ f_{X,t}, f_{A,t}\}_{t=1}^{T}$ and $f_Y$ such that 
\begin{align*}
    X_t &= f_{X,t} ( A_{t-1}, H_{t-1}, U_{X,t} ), \\
    A_t &= f_{A,t} ( H_t, U_{A,t} ) \text{, and } \\
    Y &= f_Y ( A_T, H_T, U_Y ).
\end{align*}
Here, $\Big\{ \big\{ U_{X,t}, U_{A,t}: t \in \{1, \ldots, T\} \big\}, U_Y \Big\}$ is a vector of exogenous variables. Subsequently, we'll define restrictions on their joint distribution that facilitate identification of causal effects. We will define the effects in terms of hypothetical interventions in which equation $A_t = f_{A,t}(H_t, U_{A,t})$ is removed from the structural model and the exposure is assigned as a new random variable $D_t$ (which could be deterministic). An intervention that sets exposures up to time $t-1$ to $\overline D_{t-1} \equiv \{ D_1, \dots, D_{t-1} \}$ generates counterfactual variables $X_t(\overline D_{t-1}) = f_{X,t} \big\{ D_{t-1}, H_{t-1}(\overline D_{t-2}), U_{X,t} \big\}$ and $A_t(\overline D_{t-1}) = f_{A,t} \big\{ H_t(\overline D_{t-1}), U_{A,t} \big\}$, where the counterfactual history is defined recursively as $H_t(\overline D_{t-1}) = \big\{ \overline D_{t-1}, \overline X_t(\overline D_{t-1}) \big\}$ and $A_1(D_0) = A_1$ and $X_1(D_0) = X_1$. The variable $A_t(\overline D_{t-1})$ is called the \textit{natural value of treatment} \citep{richardson2013single,   young2014identification}, and represents the possibly counterfactual value of treatment that would have been observed at time $t$ under an intervention carried out up to time $t-1$ but discontinued thereafter. An intervention in which all treatment variables up to $t=T$ are intervened on generates a counterfactual outcome $Y \left( \overline D_T \right) = f_Y \big\{ D_T, H_T(\overline D_{T-1}), U_Y \big\}$. 

\subsection{Causal assumptions}

The NPSEM implicitly encodes the consistency assumption due to the modular definition of counterfactuals within the framework, and because one subject's information does not depend on another's. This assumption would be violated if there were interference between subjects \citep{tchetgen2012causal}. We consider two exchangeability assumptions on the exogeneous variables.
\begin{assumption}[Standard sequential randomization] \label{asmp:standard-seq-exch}
    $U_{A,t} \ind \big\{ \underline{U}_{X,t+1}, U_Y \big\} \mid H_t\, \forall\, t \leq T$.
\end{assumption}
\begin{assumption}[Strong sequential randomization] \label{asmp:strong-seq-exch}
    $U_{A,t} \ind\! \{ \underline U_{X,t+1},\! \underline U_{A,t+1},\! U_Y \!\} \mid\! H_t \ \forall \ t \leq T$.
\end{assumption}
Assumption~\ref{asmp:standard-seq-exch} is standard for the identification of effects under longitudinal interventions \citep{diaz2023nonparametric}. It is satisfied if the common causes of the treatment $A_t$ and future covariates are measured.  Assumption~\ref{asmp:strong-seq-exch} is stronger. It is satisfied if common causes of treatment $A_t$ and future covariates \emph{and treatments} are measured. This assumption is similar to that required by \citet{richardson2013single} (cf. Theorem 31), and allows identification of effects under interventions that depend on the natural value of treatment \citep{young2014identification}.  We do not immediately require the typical positivity assumption, which says that time-varying propensity scores are bounded away from zero and one ($0 < \bbP(A_t = 1 \mid H_t) < 1$), because we will construct interventions which adapt to positivity violations.  Some flip interventions will require a version of positivity, and we will introduce the assumption when it is needed.

\section{Longitudinal weighting and trimming with flip interventions} \label{sec:ltt}

In this section, we extend flip interventions and interventional flip effects from Section~\ref{sec:t=1} to longitudinal data. We define longitudinal flip interventions that perform weighting or trimming on non-baseline covariates, and then develop longitudinal analogues to interventional flip effects. Depending on the weighting function, these effects remain identifiable even with arbitrary positivity violations. Importantly, the flip interventions we consider are ``single-world'', relying only on information observable under the series of interventions, which preserves their policy relevance for practical implementation.

\bigskip

To conclude the section, we investigate properties of longitudinal intervention flip effects and outline limitations of alternative estimands. Due to their single-world nature, longitudinal interventional flip effects can be driven both by differences in potential outcomes under the different target regimes and also by the impact of earlier interventions on subsequent covariates and treatments. While one might prefer estimands that isolate mechanistic differences in potential outcomes, we demonstrate that such estimands correspond to interventions with practical and interpretational challenges, justifying our focus on flip interventions.

\subsection{Longitudinal flip interventions} \label{subsec:longitudinal-flip}

We now propose flip interventions for longitudinal data. We begin by targeting a specific longitudinal regime and then define longitudinal interventional effects contrasting two flip interventions. We then establish conditions under which the resulting effects remain identifiable without standard positivity assumptions. 

\begin{definition}[Longitudinal flip interventions] 
\label{def:flip-int}
    At time $t$, suppose we are given a weight function $f_t:\mathcal{H}_t \to [0,1]$, which maps covariate and treatment history through time $t$ to $[0,1]$. Further suppose a sequence of interventions has occurred before time $t$, which we denote generically as $\overline D_{t-1}$. Then, a \emph{flip intervention} at time $t$ targeting $a_t$ is
    \[
    D_{f_t}(a_t) = 
    \begin{cases}
        A_t(\overline D_{t-1}), &\text{if } A_t(\overline D_{t-1}) = a_t, 
        \\
        a_t \one \Big[ V_t \leq f_t \{ H_t(\overline D_{t-1}) \} \Big] + A_t (\overline D_{t-1}) \one \big[ V_t > f_t \{ H_t(\overline D_{t-1}) \} \big ], &\text{otherwise,}
    \end{cases}
    \]
    where $V_t \sim \text{Unif}(0,1)$. Given a target regime $\overline a_T$ and a sequence of weight functions $\overline f = \{ f_1, \dots, f_T \}$, a sequence of flip interventions $\overline D_{\overline f}(\overline a_T)$ targeting $\overline a_T$ is defined as above, but replacing $\overline D_{t-1}$ by $\{ D_{f_1}(a_t), \dots, D_{f_{t-1}}(a_{t-1}) \}$. Moreover the sequence of interventions will depend on auxiliary random variables $\{ V_1, \dots, V_T\}$, one for each timepoint, which are iid and satisfy $\{ V_1, \dots, V_T \} \ind Z$.
\end{definition}
\noindent In words, at time $t$:
\begin{itemize}
    \item if the natural value of treatment is already $a_t$, the flip intervention does nothing; 
    \item otherwise, it flips the subject to the target treatment with probability $f_t \{ H_t(\overline D_{t-1}) \}$.
\end{itemize} 
Importantly, the intervention is ``single-world,'' meaning the decision at time $t$ depends only on information observable under interventions performed up to that point. For example, if $f_t \{ H_t(\overline D_{t-1}) \}$ is a trimming indicator, it will be a trimming indicator using the ``natural'' propensity score,  $\bbP \{ A_t(\overline D_{t-1}) = a_t \mid H_t(\overline D_{t-1}) \}$. 

\bigskip

Table~\ref{tab:long-fs} gives additional examples of weights. Once we have introduced the necessary notation to define the natural target propensity scores, constructing the weights becomes straightforward: we simply replace the usual single-timepoint propensity score in Table~\ref{tab:wates} with the natural propensity score. Additionally, if one has access to methods that achieve balance on counterfactual time-varying covariates, then the weighting function could be defined directly. To our knowledge, the approach proposed by \citet{viviano2021dynamic} is the primary example of this. 

\begin{table}[ht] 
    \centering 
    \begin{tabular}{p{5cm} p{10cm}}
        \toprule
        \textbf{Type of weighting} & \textbf{Weight/flipping probability $f_t \{ H_t(\overline D_{t-1}) \}$} \\
        \midrule
        No weighting & \( 1 \) \\
        \addlinespace
        Weighting to subjects that can take target treatment & \( \bbP \{ A_t(\overline D_{t-1}) = a_t \mid H_t(\overline D_{t-1}) \} \) \\
        \addlinespace
        Weighting to subjects that can take non-target treatment & \( 1 - \bbP \{ A_t(\overline D_{t-1}) = a_t \mid H_t(\overline D_{t-1}) \} \) \\
        \addlinespace
        Overlap weighting & \( \bbP \{ A_t(\overline D_{t-1}) = 0 \mid H_t(\overline D_{t-1}) \} \bbP \{ A_t(\overline D_{t-1}) = 1 \mid H_t(\overline D_{t-1}) \} \) \\
        \addlinespace
        Trimming & \( \one \{ \bbP \{ A_t(\overline D_{t-1}) = a_t \mid H_t(\overline D_{t-1}) \} \geq \varepsilon \} \) for $\varepsilon \geq 0$ \\
        \addlinespace
        Smooth trimming & \( s_t \big[ \bbP \{ A_t(\overline D_{t-1}) = a_t \mid H_t(\overline D_{t-1}) \} ; \varepsilon \big] \), where $s_t(x; \varepsilon)$ approximates $\one (x \geq \varepsilon)$ \\
        \addlinespace
        Matching-style weighting & \( \bbP \{ A_t(\overline D_{t-1}) = 0 \mid H_t(\overline D_{t-1}) \} \wedge \bbP \{ A_t(\overline D_{t-1}) = 1 \mid H_t(\overline D_{t-1}) \}  \) \\
        \addlinespace
        Direct covariate balancing & Varies and is data-dependent; these must be derived to directly balance \emph{counterfactual} covariates \\
        \bottomrule
    \end{tabular}
    \caption{Extending weights/flipping probabilities to longitudinal data.}
    \label{tab:long-fs}
\end{table}

\noindent We next define a longitudinal flip effect. 
\begin{definition}[Longitudinal interventional flip effect] \label{def:long-int-flip-effect}
    Let $\overline D_{\overline f}(\overline a_T)$ and $\overline D_{\overline f}(\overline a_T^\prime)$ denote two sequences of flip interventions targeting $\overline a_T$ and $\overline a_T^\prime$, respectively. Define the longitudinal interventional flip effect from these two interventions as
    \begin{equation} \label{eq:long-int-flip-effect}
        \frac{\bbE \Big[ Y \big\{ \overline D_{\overline f}(\overline a_T) \big\} - Y \big\{ \overline D_{\overline f}(\overline a_T^\prime) \big\} \Big]}{\frac{1}{T} \sum_{t=1}^{T} |\bbE \left\{  D_{f_t}(a_t) - D_{f_t}(a_t^\prime) \right\}|}.
    \end{equation}
\end{definition}

The effect defined in \eqref{eq:long-int-flip-effect} generalizes the single-timepoint interventional flip effect, characterizing the \emph{average change in potential outcomes standardized by the average absolute per-timepoint change in the number of treatments}. While the numerator matches the single-timepoint definition, multiple options exist for the denominator. We adopt a per-timepoint absolute difference approach, allowing for flexible policy contrasts where $\bbE \left\{  D_{f_t}(a_t) - D_{f_t}(a_t^\prime) \right\}$ may change signs across timepoints. In scenarios with monotonic interventions---where, without loss of generality, $D_{f_t}(a_t) \geq D_{f_t}(a_t^\prime)$ almost surely for all $t \in \{1, \dots, T\}$---the absolute value function becomes unnecessary, and one retains the same interpretation as in single-timepoint data.

\begin{remark}
    There are other options for the denominator in \eqref{eq:long-int-flip-effect}, each corresponding to different measures of distance between two longitudinal treatment distributions. We leave a full investigation to future work, but briefly note that other options include:
    \begin{itemize}
        \item \emph{Average per-timepoint probability of switching treatment assignment}: $\frac{1}{T} \sum_{t=1}^{T} \bbP \left\{ D_{f_t}(a_t) \neq D_{f_t}(a_t^\prime) \right\}$. This is an intuitive notion of the distance between the two treatment distributions, but one should note that it would not yield interventional flip effects in single-timepoint data. Nonetheless, it is a useful alternative to our proposal in \eqref{eq:long-int-flip-effect}. 
        \item \emph{Joint distributional distances}: To incorporate dependence across timepoints explicitly, one might use distances between joint distributions, such as $f$-divergences or optimal transport metrics, comparing the distributions $\mathcal{P}$ and $\mathcal{P}^\prime$ of $\overline D_{\overline f}(\overline a_T)$ and $\overline D_{\overline f}(\overline a_T^\prime)$, respectively.
    \end{itemize}
\end{remark}

\begin{remark}
    The notation we use in Definition~\ref{def:long-int-flip-effect} obscures one important nuance: the sequences of weights under each intervention are typically different. We focus on asymmetric trimming and smooth trimming, which only ensure that the propensity score for the target treatment is bounded away from zero; that is, $f_t\{ H_t(\overline D_{t-1}) \} = \one \big[ \bbP \{ A_t(\overline D_{t-1}) = a_t \mid H_t(\overline D_{t-1}) \} > \varepsilon \big]$ where $a_t$ differs between the two interventions under consideration.  This construction does not exactly generalize WATEs, which use identical weights for both interventions.  We abandon this symmetry with longitudinal data because, as shown in the next section, enforcing identical weights across longitudinal contrasts leads to impractical interventions. Instead, our approach prioritizes intervening on as many subjects as possible and pushing them towards the target regime while maintaining robustness to positivity violations. Nonetheless, the longitudinal interventional flip effect in Definition~\ref{def:long-int-flip-effect} preserves the same interpretation as a WATE: a difference in mean potential outcomes per unit treated of additional treatment (per timepoint).
\end{remark}

The next result establishes conditions under which flip interventions correspond to identifiable functionals, even under arbitrary positivity violations. We focus on the mean potential outcome under one sequence of flip interventions and the average treatment at an arbitrary timepoint $t$. 

\begin{theorem}[Longitudinal identification with flip interventions] 
\label{thm:id-flip}
    Let $\overline D_{\overline f}(\overline a_T)$ denote a sequence of flip interventions as in Definition~\ref{def:flip-int} targeting treatment regime $\overline a_T$.  Then, suppose the NPSEM and Assumption~\ref{asmp:strong-seq-exch} hold, and either of the following conditions is satisfied:
    \begin{enumerate}
        \item $\bbP \left\{ \bbP(A_t = a_t \mid H_t) = 0 \implies f_t(H_t) = 0 \right\} = 1$;  \label{cond:weight-construction}
        \item $\bbP \left\{ \bbP(A_t = a_t \mid H_t) = 0 \right\} = 0$. \label{cond:positivity}
    \end{enumerate}
    Then,
    \begin{align}
        \bbE \Big[ Y \big\{ \overline D_{\overline f}(\overline a_T) \big\}\Big] &= \sum_{\overline b_T \in \{0,1\}^T} \int_{\overline{\mathcal{X}}_T } \bbE \big(Y \mid \overline A_T = \overline b_T, \overline X_T = \overline x_T \big) \prod_{t=1}^T Q_t \bigl(b_t \mid \overline b_{t-1}, \overline x_t \bigr) d\bbP(x_t \mid \overline b_{t-1}, \overline x_{t-1}) \nonumber \\
        &= \bbE \Biggl[ Y \,\prod_{t=1}^{T} \frac{Q_t(A_t \mid H_t)}{\bbP(A_t \mid H_t)} \Biggr]
        \label{eq:id}
    \end{align}
    and
    \begin{align}
        \bbE \big\{ D_{f_t}(a_t) \big\} &= \sum_{\overline b_{t-1} \in \{0,1\}^{t-1}} \int_{\overline{\mathcal{X}}_t } Q_t(1 \mid \overline b_{t-1}, \overline x_t) d\bbP(x_t \mid \overline b_{t-1}, \overline x_{t-1}) \cdot \prod_{s=1}^{t-1} Q_s(b_s \mid \overline b_{s-1}, \overline x_s) d\bbP(x_s \mid \overline b_{s-1}, \overline x_{s-1}) \nonumber \\
        &= \bbE \Biggl[  Q_t(A_t =1 \mid H_t) \prod_{s=1}^{t-1} \frac{Q_s(A_s \mid H_s)}{\bbP(A_s \mid H_s)} \Biggr],
        \label{eq:id-tx}
    \end{align}
    where the probability of receiving the target treatment at time $t$ is
    \begin{equation} \label{eq:target-prop-score}
    Q_t(a_t \mid h_t) = \bbP(A_t = a \mid h_t) + f_t(h_t) \{ 1 - \bbP(A_t = a \mid h_t) \}.
    \end{equation}
\end{theorem}

Theorem~\ref{thm:id-flip} establishes that mean potential outcomes under flip interventions are identifiable under only strong sequential randomization and consistency. We provide g-formula identification and inverse weighting identification in \eqref{eq:id} \citep{robins1986new}. We also derive a result for the average number of treatments, in \eqref{eq:id-tx}.  The result requires that the identified analogue of the weighting function is zero when the observed propensity score for the target treatment is zero, in condition~\ref{cond:weight-construction}.  This can be enforced by construction. For example, the overlap weights, trimming and smooth trimming weights, and matching-style weights in Table~\ref{tab:long-fs} all satisfy this condition. Other weights, like Shannon's entropy weights, also satisfy this condition.  However, ``no weighting'' or weighting towards subjects that can take the non-target treatment would fail to satisfy this condition, and then it is necessary that positivity is satisfied with respect to the target treatment, in condition~\ref{cond:positivity}.

\bigskip

Before proceeding, we highlight several key observations that provide further context to these effects:

\begin{enumerate}
    \item \textbf{Complementary interpretation to incremental propensity score interventions (IPSIs).} Flip effects retain robustness to arbitrary positivity violations, making them a complementary alternative to IPSIs \citep{kennedy2019nonparametric, bonvini2023incremental}. However, the interventions have a different construction and thus a different interpretation. IPSIs multiply the odds ratio of treatment, whereas flip interventions target a specific treatment and only adapt to positivity violations via the weight function. In this way, IPSIs are a natural smooth extrapolation from the observed data distribution, whereas flip interventions are tailored to target a specific counterfactual regime directly and only adapt to positivity violations as required. Depending on the context, either could be preferable in practice.

    \item \textbf{Interventions depending on the natural value of treatment.} A critique of interventions that depend on the \textit{natural} value of treatment is that they may be impractical because this value is unobserved in practice. This issue can be addressed in two ways:
    
    \begin{enumerate}[label=(\roman*)]
        \item An approximation can be constructed by defining interventions based on a subject's intended treatment, which may closely approximate their natural treatment value. See \citet[Section 6]{young2014identification} for a discussion.
        \item It is possible to define flip interventions that do not depend on the natural value of treatment while still yielding the same identification result as Theorem~\ref{thm:id-flip}. The next point elaborates on this modification.
    \end{enumerate}

    \item \textbf{Relaxing the sequential randomization assumption.} The identification result in Theorem~\ref{thm:id-flip} relies on strong sequential randomization (Assumption~\ref{asmp:strong-seq-exch}) because flip interventions depend on the natural treatment value to retain an intuitive ``flipping'' interpretation. However, this assumption can be relaxed to standard sequential randomization (Assumption~\ref{asmp:standard-seq-exch}) by redefining the interventions so they do not depend on the natural treatment value. Specifically, one could instead define
    \begin{align*}
        D_{f_t}(a_t) = \one \Big( V_t \leq &\; \one(a_t = 1) f_t\{ H_t(\overline D_{t-1}) \} \\
        &+ \Big[ 1 - f_t\{ H_t(\overline D_{t-1}) \} \Big] \bbP\{ A_t(\overline D_{t-1}) =1 \mid H_t(\overline D_{t-1}) \} \Big) 
    \end{align*}
    where \( V_t \sim \text{Unif}(0,1) \). These redefined interventions satisfy the identification result in \eqref{eq:id} under standard sequential randomization and the identification result for the average number of treatments under only the consistency assumption embedded in the NPSEM. Moreover, they do not suffer from the practical concerns discussed in the previous point. 

    \item \textbf{Connections to maximally coupled policies.} Flip interventions that depend on the natural value of treatment are special cases of ``maximally coupled generalized policies'' applied to binary treatments, where the intervention produces a non-decreasing shift in propensity scores for the target treatment \citep{levis2024stochastic}. These policies  minimize the number of subjects intervened on while preserving a target intervention propensity score, \( Q_t(A_t \mid H_t) \).  This approach was originally proposed to minimize bounds on causal effects under unmeasured confounding (e.g., adapting IPSIs \citep[Section 3.3]{levis2024stochastic}). Here, we repurpose these interventions because they have a nice interpretation as flip interventions. Examining their robustness to unmeasured confounding remains an open question for future work.
\end{enumerate}

\subsection{Properties of longitudinal interventional flip effects; drawbacks of direct weighting and trimming} \label{subsec:properties}

In this section, we investigate the properties of longitudinal interventional flip effects and outline issues with alternative estimands one might consider, justifying our focus on flip interventions. First, we note that longitudinal interventional flip effects satisfy a minimal sharp null preservation property: if the treatment has no effect on the outcome but the interventions shift the treatment distribution, then the interventional flip effect is zero. 
\begin{proposition}
    Let $\overline D_{\overline f}(\overline a_T)$ and $\overline D_{\overline f}(\overline a_T^\prime)$ denote two flip interventions. If $Y(\overline b_T) = Y(\overline b_T^\prime)$ for all $\overline b_T, \overline b_T^\prime \in \{0,1\}^T$ but $\sum_{t=1}^{T} | \bbE \{ D_{f_t}(a_t) \} - \bbE \{ D_{f_t}(a_t^\prime) \}| > 0$, then 
    \[
    \frac{\bbE \left[ Y \left\{ \overline D_{\overline f}(\overline a_T) \right\} - Y \left\{ \overline D_{\overline f}(\overline a_T^\prime) \right\} \right]}{T^{-1} \sum_{t=1}^{T} |\bbE \{ D_{f_t}(a_t) \} - \bbE \{ D_{f_t}(a_t^\prime) \}|} = 0.
    \]
\end{proposition}
Although longitudinal interventional flip effects satisfy the minimal sharp null criterion, they present interpretive challenges that distinguish them from their single-timepoint counterparts. In the single-timepoint case, interventional flip effects have an intuitive interpretation as a conditional effect amongst the flipped population: $\psi_f = \bbE \{ Y(1) - Y(0) \mid V \leq f(X) \}$. Indeed, when the weight $f(X)$ is a trimming indicator, this simplifies to a conditional trimmed effect. Unfortunately, this clear interpretation does not generalize to longitudinal data. Intuitively, this generalization fails because of the dynamic nature of the sequence of interventions: the intervention at timepoint $t$ affects subsequent natural covariates and treatment as well as the ultimate outcome of interest.

\bigskip

This feature of longitudinal interventional flip effects serves as the basis of a critique of direct weighting or trimming methods that isolate a mean difference in potential outcomes in a specific sub-population. We will focus on trimming for simplicity, but similar arguments can be made about direct weighting. When conducting direct trimming on propensity scores in longitudinal settings, researchers presumably have in mind an estimand that isolates the difference $Y(\overline a_T) - Y(\overline a_T^\prime)$ amongst individuals with propensity scores above a certain threshold under both treatment regimes. This ideal estimand can be written as: \small
\begin{equation} \label{eq:cross-world}
    \hspace{-0.2in}\bbE \left[ Y(\overline a_T) - Y(\overline a_T^\prime) \ \Big| \ \bigcap_{t=1}^{T} \Bigl\{ \bbP \bigl\{ A_t(\overline a_{t-1}) = a_t \mid H_t(\overline a_{t-1}) \bigr\} > \varepsilon,\; \bbP \bigl\{ A_t(\overline a_{t-1}^\prime) = a_t^\prime \mid H_t(\overline a_{t-1}^\prime) \bigr\} > \varepsilon \Bigr\} \right]
\end{equation} \normalsize
While this estimand has the appealing property of isolating the difference $Y(\overline a_T) - Y(\overline a_T^\prime)$ among a well-defined subpopulation, it suffers from two major limitations.   First, it is ``cross-world.'' Notice that the conditioning set depends on counterfactual covariates under two treatment regimes. Consequently, an intervention corresponding to this effect cannot be implemented as a single-world intervention and it cannot be falsified experimentally or implemented in practice. This limitation parallels natural effects in mediation \citep{andrews2021insights, richardson2013single}.  Second, the effect corresponds to a contrast under \emph{future-dependent} interventions.  This is established in the next result.

\begin{proposition} \label{prop:simultaneous}
    Let \(  \Pi_T := \prod_{t=1}^{T} \one \Bigl[
   \bbP \bigl\{ A_t(\overline a_{t-1}) = a_t \mid H_t(\overline a_{t-1}) \bigr\} > \varepsilon,\; \bbP \bigl\{ A_t(\overline a_{t-1}^\prime) = a_t^\prime \mid H_t(\overline a_{t-1}^\prime) \bigr\} > \varepsilon \Bigr] \). Then, the treatment decisions 
    \begin{itemize}
        \itemsep0in
        \item \( \overline D_T = \one(\overline A_T = \overline a_T) \overline A_T + \one(\overline A_T \neq \overline a_T) \Big\{ \overline a_T \one \left( V \leq \Pi_T \right) + \overline A_T \one \left( V > \Pi_T \right) \Big\} \) and
        \item \( \overline D_T^\prime = \one(\overline A_T = \overline a_T^\prime) \overline A_T + \one(\overline A_T \neq \overline a_T^\prime) \Big\{ \overline a_T^\prime \one \left( V^\prime \leq \Pi_T \right) + \overline A_T \one \left(V^\prime >  \Pi_T \right) \Big\} \) 
    \end{itemize}
    yield a longitudinal interventional flip effect equal to the trimmed treatment effect in \eqref{eq:cross-world}.
\end{proposition}

Proposition~\ref{prop:simultaneous} shows that the trimmed effect in \eqref{eq:cross-world} is defined using \emph{simultaneous} flip interventions: for subjects that would have followed the target regime, there is no intervention; otherwise, the intervention flips their treatment to the target regime with probability equal to the weighted product across all timepoints. A key limitation of these effects arises from the nature of these interventions, which \emph{simultaneously} alter the entire regime and rely on \emph{future information} at earlier timepoints. For instance, the intervention at the first timepoint depends on a subject’s natural treatment and covariate values at all timepoints. The simultaneous and cross-world nature of the intervention hampers the interpretability and practicality of the effects. These issues motivate our focus on flip interventions with longitudinal data and longitudinal interventional flip effects.

\section{Estimation and inference} \label{sec:efficiency}

In this section, we outline methods for estimating flip effects. We focus on estimating mean potential outcomes $\bbE \big[ Y \{ \overline D_{\overline f}(\overline a_T) \} \big\}$. The methods will also apply to estimating the average number of treatments $\bbE \{D_{f_t}(a_t) \}$, and we provide full details in Appendix~\ref{app:trt_prob}. Estimating and conducting inference on longitudinal interventional flip effects will follow by the delta method under mild regularity conditions. 

\bigskip

Throughout, we have assumed that the weight function was known a priori. We will continue to do so in this section. When this is not the case --- for example, if one wanted to decide a trimming threshold or smooth trimming parameter data-adaptively --- then estimation and inference are more complex; see \citet{khan2022doubly} for a review. 

\bigskip

Moreover, we will assume the weight function is a smooth function of the target propensity score. Specifically, we will assume the identified weight function $f_t(H_t)$ in Theorem~\ref{thm:id-flip} satisfies
\[
f_t(H_t) = s_t\{ \bbP(A_t=a_t \mid H_t) \} 
\]
where $s_t(\cdot)$ is twice differentiable with non-zero and bounded derivatives. This includes all the examples in Table~\ref{tab:long-fs} except trimming and matching-style weighting.  The smoothness of the weight function is crucial to allow for $\sqrt{n}$-convergence under nonparametric conditions.  With non-smooth weights, such as the trimming indicator, the lack of pathwise differentiability due to the non-smoothness of the weight function creates complications. Without additional assumptions, the performance of estimators for these effects is dictated by the behavior of propensity score estimators within the trimming indicator. While $\sqrt{n}$-rate estimation or valid inference may be possible under parametric models for the propensity scores or with specific nonparametric assumptions and estimators, general guarantees are unavailable. Note that if trimming or matching-style weights are desired, they can be replaced by smooth approximations, as we explore in our data application in Section~\ref{sec:data-analysis}. In any case, for the reasons mentioned, we will focus our theoretical development on smooth weights, which are pathwise differentiable and allow for the construction of $\sqrt{n}$-consistent and asymptotically normal estimators under nonparametric assumptions by leveraging nonparametric efficiency theory and efficient influence functions \citep{bickel1993efficient}. We will first establish the efficient influence function for the flip effect and then we will use it to construct multiply robust and sequentially doubly robust estimators.

\begin{remark}
    There is an assumption-lean approach yielding $\sqrt{n}$-estimation with a non-smooth weight function. By targeting data-dependent weights $\widehat f$, rather than the true underlying weights, one can obtain $\sqrt{n}$-rates for the corresponding data-dependent flip effect. \citet{van2007causal} demonstrates this approach with ``realistic'' interventions. Similarly, we could target a data-dependent flip effect and attain $\sqrt{n}$-rates with simple modifications of existing methods.
\end{remark}

\subsection{Notation}

To facilitate exposition, we refine our notation. First, we let  
\begin{equation} \label{eq:ratio}
    r_t(b_t \mid h_t) = \frac{Q_t(b_t \mid h_t)}{\bbP(A_t = b_t \mid h_t)} 
\end{equation}
be the ratio of the intervention propensity score, defined in~\eqref{eq:target-prop-score}, and the true propensity score, and we define $r_0 = 1$ and $Q_{T+1}(A_{T+1} \mid H_{T+1}) = 1$. Then, we let $m_{T+1} = Y$, $m_T(b_T, H_T) = \bbE(Y \mid A_T = b_T, H_T)$, and recursively define
\begin{equation} \label{eq:seq-reg}
    m_t(b_t, h_t) = \bbE \left\{ \sum_{b_{t+1}} m_{t+1}(b_{t+1}, H_{t+1}) Q_{t+1}(b_{t+1} \mid H_{t+1})  \,\, \bigg| \,\, A_t = b_t, H_t = h_t \right\}
\end{equation}
as the sequential regression function for $t < T$.  Note that $m_0 = \psi \equiv \bbE \left[ Y \left\{ \overline D_T\left( \overline a_T \right) \right\} \right]$.

\subsection{Efficient influence function}

The identification result in Theorem~\ref{thm:id-flip} suggests a ``plug-in estimator'' by plugging estimates of the relevant nuisance functions into each of the relevant formulas and then taking a sample average.  With well-specified  parametric models for the nuisance functions, the plug-in estimator can achieve $\sqrt{n}$-convergence rates. However, if the models are mis-specified, the plug-in estimator can be biased \citep{vansteelandt2012model, kang2007demystifying}. Meanwhile, if the nuisance functions are estimated with nonparametric methods, the plug-in estimator will  typically inherit slower-than-$\sqrt{n}$ nonparametric convergence rates. This motivates estimators based on nonparametric efficiency theory \citep{bickel1993efficient, van2000asymptotic, tsiatis2006semiparametric}.

\bigskip

The first-order bias of the plug-in estimator with nonparametric data-adaptive regression estimators can be characterized by the efficient influence function of the functional, which can be thought of as its first derivative in a von Mises expansion \citep{von1947asymptotic}. The efficient influence function can be used to construct estimators that can achieve $\sqrt{n}$-convergence with flexible, nonparametric estimators for the nuisance functions.  The next result establishes the efficient influence function of $\bbE \big[ Y\big\{ \overline D_{\overline f}(\overline a_T) \big\} \big]$.

\begin{proposition} \label{prop:eif}
    Let $\psi$ denote an identified flip effect $\bbE \big[ Y\big\{ \overline D_{\overline f}(\overline a_T) \big\} \big]$ from Theorem~\ref{thm:id-flip} with smooth weight function. Moreover, let 
    \begin{align*}
        \phi_t(b_t; A_t, H_t) = &\ \Big\{ 2 \one(b_t = a_t) - 1 \Big\} \Big\{ \one(A_t = a_t) - \bbP(A_t = a_t \mid H_t) \Big\} \\
        &\cdot \Big[ 1 - s_t\{ \bbP(A_t = a_t \mid H_t) \} + s_t^\prime \{ \bbP(A_t = a_t \mid H_t) \} \big\{ 1 - \bbP(A_t = a_t \mid H_t) \big\} \Big]
    \end{align*}
    where $s_t^\prime(y) = \frac{\partial}{\partial x} s_t(x) \big|_{x=y}$. Further suppose that the outcome $Y$ has bounded variance and the weight function is constructed such that $r_t(A_t \mid H_t)$ is uniformly bounded. Then, the efficient influence function of $\psi$ under a nonparametric model is
    \begin{align*}
        \varphi(Z) &= \varphi_m(Z) + \varphi_Q(Z) \text{ where} \nonumber \\
        \varphi_m(Z) &= \sum_{t=0}^{T} \left\{ \prod_{s=0}^{t} r_s(A_s \mid H_s) \right\} \left\{ \sum_{b_{t+1}} m_{t+1}(b_{t+1}, H_{t+1}) Q_{t+1}(b_{t+1} \mid H_{t+1}) - m_t(A_t, H_t) \right\},  \\
        \varphi_Q(Z) &= \sum_{t=1}^{T} \left\{ \prod_{s=1}^{t-1} r_s(A_s \mid H_s) \right\} \sum_{b_t} m_t(b_t, H_t) \phi_t(b_t; A_t, H_t).
    \end{align*}
\end{proposition}

The efficient influence function in Proposition~\ref{prop:eif} follows a typical structure: $\varphi(Z)$ consists of a plug-in estimator minus the true functional plus weighted residuals. The first component, $\varphi_m(Z)$, represents the efficient influence function that would arise if $Q_t(A_t \mid H_t)$ were known and did not require estimation. The second component, $\varphi_Q(Z)$, emerges from the necessity of estimating this quantity. It includes $\phi_t(b_t; A_t, H_t)$, which is the  efficient influence function of $\bbE \{ Q_t(A_t = b_t \mid H_t) \}$.

\bigskip

The result requires bounded variance of $\varphi(Z)$, which is guaranteed if $Y$ has bounded variance and $r_t(A_t \mid H_t)$ is bounded for all $t \leq T$. The boundedness condition on $r_t$ can be guaranteed through appropriate construction of the weight function. All the smooth weights in Table~\ref{tab:long-fs} except for the smooth trimming weights satisfy it immediately. For smooth trimming weights, the bound can be satisfied by construction. For instance, choosing $s(x) = 1 - \exp(-kx)$ ensures $r_t(A_t \mid H_t) = \frac{Q_t(A_t \mid H_t)}{\bbP(A_t \mid H_t)}$ is bounded since $s(x) / x \leq k$.

\subsection{Multiply robust-style estimator}

The efficient influence function in Proposition~\ref{prop:eif} suggests a multiply robust-style estimator. For simplicity in the analysis, we'll assume there are $2n$ observations and use a sample split estimator, and let $\bbP_n$ denote a sample average over the evaluation data. 

\begin{algorithm}[Multiply robust-style estimator] \label{alg:mult-dr-est} 
    Assume training and evaluation datasets of $n$ observations. 
    \begin{enumerate}
        \item For all timepoints, regress $A_t$ on $H_t$ in the training data and obtain propensity score models; using these models, compute the intervention propensity scores $\widehat Q(A_t \mid H_t)$, ratios $\widehat r_{t}(A_t \mid H_t)$, and efficient influence functions $\widehat \phi_{t}(b_t; A_t, H_t)$ for all samples and timepoints.
        \item For $t=T$ to $t=1$: 
        \begin{enumerate}
            \item 
            \begin{enumerate}
                \item If $t = T$, then $\widehat P_{T+1}(H_{T+1}) = Y$. Otherwise, pseudo-outcome $\widehat P_{t+1}(H_{t+1})$ is available from the previous step in loop (see step 3 below).
                \item In the training data, regress $\widehat P_{t+1}(H_{t+1})$ against $A_t$ and $H_t$ to obtain a sequential regression model; using this model, across the full data, obtain estimates $\widehat m_{t}(0, H_t), \widehat m_{t}(1, H_t)$.
            \end{enumerate}
            \item Across the full data, compute pseudo-outcomes \\ $\widehat P_t(H_t) = \widehat m_t(0, H_t) \widehat Q_t(0 \mid H_t) + \widehat m_t(1, H_t) \widehat Q_t(1 \mid H_t)$ to use in the next step.
        \end{enumerate}
        \item In the evaluation data only, compute the plug-in estimator $\widehat m_0 = \bbP_n\{\widehat P_1(X_1) \}$ and efficient influence function values by plugging nuisance estimates into $\varphi(Z)$ in Proposition~\ref{prop:eif}.
    \end{enumerate}
    Finally, output the point estimate and variance estimate
    \[
    \widehat \psi := \widehat m_0 + \bbP_n \{ \widehat \varphi(Z) \} \text{ and } \widehat \sigma^2 := \bbP_n \left\{ \widehat \varphi(Z)^2 \right\}.
    \]
\end{algorithm}

Algorithm~\ref{alg:mult-dr-est} constructs an estimate of the efficient influence function by first estimating $\{ \widehat Q_t \}_{t=1}^T$ and then working sequentially from $t=T$ to $t=1$ to estimate $\{ \widehat m_t \}_{t=1}^T$. This sequential regression formulation is the same as in \citet{kennedy2019nonparametric}, and uses an estimated pseudo-outcome $\widehat P_{t+1}(H_{t+1})$ in a regression to estimate $m_t(A_t, H_t)$. An alternative is the targeted maximum likelihood estimator (TMLE) in \citet{diaz2023nonparametric}, which offers the same asymptotic guarantees. Algorithm~\ref{alg:mult-dr-est} also employs sample splitting and cross-fitting to avoid relying on Donsker or other complexity conditions on the nuisance function estimators \citep{robins2008higher, chernozhukov2018double, zheng2010asymptotic, chen2022debiased, van1996weak}. Therefore, we are agnostic to the choice of regression method.  To retain full-sample efficiency, one could cycle the folds in the estimator above, repeat, and average; the same error guarantees apply \citep{diaz2023nonparametric}.

\subsubsection{Multiply robust-style convergence guarantees}

The next result provides the primary convergence guarantee for this estimator: a bound on its bias. We then show that the estimator satisfies a rate multiply robust-style result, in the sense of \citet{rotnitzky2021characterization}, describing when $\sqrt{n}$-efficiency and asymptotic normality hold. 

\begin{theorem} \label{thm:flip-bias}
    Under the setup of Proposition~\ref{prop:eif}, let $\widehat \psi$ denote a point estimate from Algorithm~\ref{alg:mult-dr-est} and let
    \begin{itemize}
        \item $\widetilde m_t(A_t, H_t) = \bbE \left\{ \sum_{b_{t+1}} \widehat m_{t+1}(b_{t+1}, H_{t+1}) \widehat Q_{t+1}(b_{t+1} \mid H_{t+1}) \, \middle| \, A_t, H_t \right\}$, for $t < T$, and
        \item $\widehat \pi_t(H_t) := \widehat \bbP(A_t = 1 \mid H_t)$, for $t \leq T$.
    \end{itemize}
    Suppose $\exists\ C < \infty$ such that $\bbP \big\{ \widehat m_t(A_t, H_t) \leq C \big\} = \bbP \big\{ m_t (A_t, H_t) \leq C \big\} = 1$ for $t \leq T$. Then,
    \begin{align*}
        \left| \bbE \left( \widehat \psi - \psi \right) \right| \lesssim \min \bigg\{ &\sum_{t=1}^{T} \| \widehat \pi_t - \pi_t \| \| \widehat m_t - m_t \| + \| \widehat \pi_t - \pi_t \|^2, \\
        &\sum_{t=1}^{T}  \| \widehat m_t - \widetilde m_t \|  \Big( \sum_{s=1}^{t} \| \widehat \pi_s -\pi_s \| \Big) + \| \widehat \pi_t - \pi_t \| \Big( \sum_{s=1}^{t} \| \widehat \pi_s -\pi_s \| \Big) \bigg\}.
    \end{align*}
\end{theorem} 

Theorem~\ref{thm:flip-bias} provides a bound on the bias of the multiply robust-style estimator. Under the assumptions of Proposition~\ref{prop:eif}, we only require that both the true and estimated regression functions $m_t$ and $\widehat{m}_t$ are bounded. This result provides three new contributions in longitudinal data and sheds light on estimating WATEs in single-timepoint data:

\begin{enumerate}[itemsep=0.05in]
    \item \textbf{Simultaneous bounds on the bias.} We establish that two bounds hold at once, so the bias can be bounded by their minimum. To our knowledge, this is novel. This arises because the total bias decomposes into a sum of errors from $t=1$ to $t=T$, with the first part of the minimum obtained by decomposing the error at future timepoints via the sequential regression $\widehat m_t$, and the second by decomposing the error at past timepoints via $\{ \widehat Q_s \}_{s \leq t}$.
    
    \item \textbf{Extension of \citet[Theorem 3]{diaz2023nonparametric}.} The first part of the minimum extends \citet[Theorem~3]{diaz2023nonparametric} to a one-step estimator and to stochastic LMTPs with unknown $Q_t$. In this setting, the dependence on future timepoints $s \ge t$ is contained in $\|\widehat{m}_t - m_t\|$, which captures the errors from sequential regressions from $s=T$ to $s=t$, as well as from the propensity scores $\{\widehat{Q}_s\}_{s>t}$ that define the pseudo-outcomes. Our bound is new in explicitly incorporating $\|\widehat{\pi}_t - \pi_t\|^2$, reflecting the fact that the flipping probabilities must be estimated.
    
    \item \textbf{A tighter bound than \citet[Theorem 3]{kennedy2019nonparametric}.} The second part of our bound depends only on the sequential regression error at time~$t$, ignoring pseudo-outcome estimation. Specifically, it involves $\|\widehat{m}_t - \widetilde{m}_t\|$, whereas \citet[Theorem 3]{kennedy2019nonparametric} upper bounds the same term by $\|\widehat{m}_t - m_t\|$, which implicitly includes additional error from future propensity scores and sequential regressions (as discussed in point 2.).

    \item \textbf{Doubly robust-style bounds for WATE estimation when $T=1$.} For single-timepoint data, this estimator yields doubly robust-style bounds for estimating the class of WATEs described in Section~\ref{sec:t=1}. In some cases, this was already known; for example, it is well-established that one can upper bound the bias of an estimator for the ATO by $\sum_{a\in \{0,1\}} \| \widehat \mu_a - \mu_a \| \| \widehat \pi - \pi \| + \| \widehat \pi - \pi\|^2$, where $\mu_a = \bbE(Y \mid A=a, X)$, because the ATO can be re-written as $\tfrac{\bbE \left\{ \cov (A, Y \mid X) \right\}}{\bbE \left\{ \bbV(A \mid X)  \right\}}$. Our result includes that bias bound as a special case. For more complex smooth weight functions, such as Shannon's entropy weights---which take the form $f(X) = - \big[ \pi(X) \log \pi(X) + \{ 1 - \pi(X) \} \log \{ 1 - \pi(X) \} \big]$---the literature has not, to our knowledge, developed doubly robust-style estimators that allow for nonparametric nuisance estimation while accounting for uncertainty in propensity score estimation during weight construction. Our result does so.
\end{enumerate}

\noindent This bound on the bias indicates when weak convergence is possible.

\begin{corollary}[Multiple robustness and weak convergence] 
\label{cor:s-flip-convergence}
    Under the setup of Theorem~\ref{thm:flip-bias}, let $\widehat{\sigma}^2$ be a variance estimate from Algorithm~\ref{alg:mult-dr-est}. Suppose $\lVert \widehat{\varphi} - \varphi \rVert = o_{\mathbb{P}}(1)$ and 
    \begin{align*}
        \min \bigg\{ &\sum_{t=1}^{T} \| \widehat \pi_t - \pi_t \| \| \widehat m_t - m_t \| + \| \widehat \pi_t - \pi_t \|^2, \\
        &\sum_{t=1}^{T}  \| \widehat m_t - \widetilde m_t \|  \Big( \sum_{s=1}^{t} \| \widehat \pi_s -\pi_s \| \Big) + \| \widehat \pi_t - \pi_t \| \Big( \sum_{s=1}^{t} \| \widehat \pi_s -\pi_s \| \Big) \bigg\} = o_\bbP(n^{-1/2}).
    \end{align*}
    Then,
    \begin{equation}
        \sqrt{\frac{n}{\widehat{\sigma}^2}}\bigl(\widehat{\psi} - \psi\bigr) 
        \indist N(0, 1).
    \end{equation}
\end{corollary}

Corollary~\ref{cor:s-flip-convergence} provides a multiply robust-style guarantee, showing conditions under which the estimator achieves $\sqrt{n}$-convergence to a Gaussian limit. Specifically, the first requirement ensures that the estimated efficient influence function is consistent, and the second is the crucial multiply robust-style bound on the bias. In particular, the product of the nuisance estimation errors from Theorem~\ref{thm:flip-bias} must converge to zero at a rate of $n^{-1/2}$. This condition is achievable under nonparametric assumptions on the nuisance functions (e.g., smoothness, sparsity, or bounded variation), where each nuisance function can be estimated at a $n^{-1/4}$ rate \citep{gyorfi2002distribution}.

\begin{remark}
    When $Q_t$ is unknown, a \emph{model} multiply robust-style result for consistency (in the sense of \citet{rotnitzky2021characterization}) is less immediately interesting than in settings with known $Q_t$. When $Q_t$ is unknown, consistent estimation of $\{\pi_t\}_{t=1}^T$ is necessary, but is also sufficient, to guarantee $\widehat{\psi} \xrightarrow{\mathbb{P}} \psi$.  However, Theorem~\ref{thm:flip-bias} implies a new result when $Q_t$ is known: $2(T+1)$ model multiple robustness; see, e.g., \citet[Lemma 2]{diaz2023nonparametric} for details on $T+1$ model multiple robustness. In other words, our result implies that the typical multiply robust estimator is \emph{twice as robust as was previously realized}.
\end{remark}

\subsection{Sequentially doubly robust-style estimator}

The multiply robust-style estimator can be improved to a sequentially doubly robust-style estimator. One can gain intuition for how this is possible by examining the estimated pseudo-outcome $\widehat P_{t+1}(H_{t+1})$ in Algorithm~\ref{alg:mult-dr-est}: regressing $\widehat P_{t+1}(H_{t+1}) = \widehat m_{t+1}(0, H_{t+1}) \widehat Q_t(0 \mid H_{t+1}) + \widehat m_{t+1}(1, H_{t+1}) \widehat Q_{t+1}(1 \mid H_{t+1})$ against $\{ A_t, H_t \}$ corresponds to using a plug-in estimator for $m_t(A_t, H_t)$.  This estimator can be improved by debiasing this pseudo-outcome. For sequential regressions with longitudinal data, this was first observed in \citet{luedtke2017sequential} and \citet{rotnitzky2017multiply}, and recently extended to LMTPs in \citet{diaz2023nonparametric}. This general approach --- debiasing a pseudo-outcome --- has also been applied to conditional effect estimation, continuous dose-response curve estimation, and censoring \citep{kennedy2017non, kennedy2023towards, mcclean2024nonparametric, rubin2007doubly}. An adaptation of the estimator in Algorithm~\ref{alg:mult-dr-est} is inspired by the following lemma.
\begin{lemma} \label{lem:seq-property}
    Under the setup of Proposition~\ref{prop:eif}, define
    $Y = m_{T+1} = \sum_{b_{T+1}} m_{T+1} \left( Q_{T+1}  + \phi_{T+1} \right)$ and recursively define for $t=T$ to $t=1$ \footnotesize
    \begin{align*}
        &P_t^\ast(Z) = \sum_{b_t} m_t (b_t, H_t) \Big\{ Q_t(b_t \mid H_t) + \phi_t(b_t; A_t, H_t) \Big\} \\
        &+ \sum_{s=t}^{T} \left\{ \prod_{k=t}^s r_k(A_k \mid H_k) \right\} \left\{ \sum_{b_{s+1}} m_{s+1} (b_{s+1}, H_{s+1}) \Big\{ Q_{s+1}(b_{s+1} \mid H_{s+1}) + \phi_{s+1}(b_{s+1}; A_{s+1}, H_{s+1}) \Big\} - m_s(A_s, H_s) \right\}.
    \end{align*}
    \normalsize Then,
    \begin{equation} \label{eq:pseudo-unbiased}
        \bbE \left\{ P_{t+1}^\ast(Z) \mid A_t, H_t \right\} = m_t(A_t, H_t).
    \end{equation}
    Moreover, suppose access to fixed nuisance estimates $\left\{ \widehat m_s^\ast, \widehat Q_s \right\}_{s=t+1}^{T}$ to construct $\widehat P_{t+1}^\ast(Z)$. Then, 
    \begin{align}
        &\bbE \left\{ \widehat P_{t+1}^\ast (Z) - m_t(A_t, H_t) \mid A_t, H_t \right\} = \nonumber \\
        &\hspace{0.2in}\sum_{s=t+1}^{T} \bbE\left[ \left\{ \prod_{k=t+1}^{s-1} \widehat r_k(A_k \mid H_k) \right\} \Big\{ m_s(A_s, H_s) - \widehat m_s^\ast(A_s, H_s) \Big\} \Big\{ \widehat r_s(A_s \mid H_s) - r_s (A_s \mid H_s) \Big\} \mid A_t, H_t \right] \nonumber \\
        &+ \sum_{s=t+1}^{T} \bbE \left[ \left\{ \prod_{k=t+1}^{s-1} \widehat r_k(A_k \mid H_k) \right\} \sum_{b_s} \widehat m_s^\ast(A_s, H_s) \left\{ \widehat Q_s(b_s \mid H_s) + \widehat \phi_s(b_s; A_s, H_s) - Q_s(b_s \mid H_s) \right\} \mid A_t, H_t \right]. \label{eq:pseudo-dr}
    \end{align}
\end{lemma}
This lemma proposes the debiased pseudo-outcome, $P_t^\ast$, then shows that it is indeed unbiased (in \eqref{eq:pseudo-unbiased}) and that its error, if it were estimated, is a product of errors (in \eqref{eq:pseudo-dr}). This mirrors Lemma 1 in \citet{diaz2023nonparametric}, Lemma 1 in \citet{luedtke2017sequential}, and Lemma 2 in \citet{rotnitzky2017multiply}. However, this result is new because it accounts for the error in estimating the intervention propensity score $Q_t$, from which the second term in the bias decomposition in \eqref{eq:pseudo-dr} arises. This result inspires a new sequentially doubly robust-style estimator for S-LMTPs, which amends Algorithm~\ref{alg:mult-dr-est}.
\begin{algorithm}[Sequentially doubly robust-style estimator] \label{alg:seq-dr-est}
    Use Algorithm~\ref{alg:mult-dr-est} with the following amendments to the sequential regression loop:
    \begin{itemize}[itemsep=0.03in]
        \item In step 2(a)(i), let $\widehat P_{T+1}^\ast(Z) = Y$.
        \item In step 2(a)(ii), in the training data regress $\widehat P_{t+1}^\ast(Z)$ against $A_t$ and $H_t$, and denote the estimates as $\widehat m_{t}^\ast(0, H_t), \widehat m_{t}^\ast(1, H_t)$.
        \item In step 3, when constructing pseudo-outcomes, use the transformation $\widehat P_t^\ast(Z)$ which uses available nuisance estimates $\left\{ \widehat m_s^\ast, \widehat Q_s \right\}_{s=t+1}^{T}$.
    \end{itemize}
    Finally, construct a point estimate and variance estimate as
    \[
    \widehat \psi^\ast = \bbP_n \{ \widehat P_1^\ast(Z) \} \text{ and } \widehat \sigma^2 := \bbP_n \Big[ \big\{ \widehat P_1^\ast(Z) - \widehat \psi^\ast \big\}^2 \Big].
    \]
\end{algorithm}
\noindent The estimator is similar to the multiply robust estimator in Algorithm~\ref{alg:mult-dr-est}, but uses the debiased pseudo-outcomes and debiased sequential regression estimates. A consequence of this is that $\widehat P_1^\ast(Z)$ already takes the same form as the un-centered efficient influence function from Proposition~\ref{prop:eif} and the point estimate and variance estimate can be constructed using $\widehat P_1^\ast(Z)$, rather than constructing an estimate of the efficient influence function.

\subsubsection{Sequentially doubly robust-style convergence guarantees}

The next result gives the sequentially doubly robust-style properties of the estimator.
\begin{theorem} \label{thm:seq-dr-bias}
    Under the setup of Theorem~\ref{thm:flip-bias}, let $\widehat \psi^\ast$ denote a point estimate from Algorithm~\ref{alg:seq-dr-est} and let $\widetilde m_t^\ast (A_t, H_t) = \bbE \left\{ \widehat P_{t+1}^\ast(Z) \mid A_t, H_t \right\}$. Moreover, suppose $\exists \ C < \infty$ such that $\bbP \{ \widehat m_t^\ast (A_t, H_t) \leq C \} = 1$ for $t \leq T$. Then,
    \begin{align*}
        \left| \bbE \left( \widehat \psi^\ast - \psi \right) \right| \lesssim \sum_{t=1}^{T} \| \widehat \pi_t - \pi_t \| \Big( \| \widehat m_t^\ast - \widetilde m_t^\ast \| + \| \widehat \pi_t - \pi_t \| \Big).
    \end{align*}
\end{theorem}
Theorem~\ref{thm:seq-dr-bias} shows that the estimate is sequentially doubly robust-style: its bias can decomposed as a sum of errors across timepoints where the error at each timepoint only depends on the propensity score at that timepoint and the sequential regression estimate at that timepoint. Note that $\| \widehat m_t^\ast - \widetilde m_t^\ast \|$ only captures the error from the sequential regression at $t$; there is no dependence on $s > t$ through the pseudo-outcome. Therefore, we have the following asymptotic convergence guarantee.
\begin{corollary} \label{cor:s-flip-seq-convergence}
    Under the setup of Theorem~\ref{thm:seq-dr-bias}, let $\widehat{\sigma}^2$ be a variance estimate from Algorithm~\ref{alg:seq-dr-est}. Suppose $\lVert \widehat P_1^\ast - P_1^\ast \| = o_{\mathbb{P}}(1)$ and 
    \[
    \sum_{t=1}^{T} \| \widehat \pi_t - \pi_t \| \Big( \| \widehat m_t^\ast - \widetilde m_t^\ast \| + \| \widehat \pi_t - \pi_t \| \Big) = o_{\mathbb{P}}(n^{-1/2}).
    \]
    Then,
    \begin{equation}
        \sqrt{\frac{n}{\widehat{\sigma}^2}}\bigl(\widehat{\psi} - \psi\bigr) 
        \indist N(0, 1).
    \end{equation}
\end{corollary}

Corollary~\ref{cor:s-flip-seq-convergence} provides a sequentially doubly robust-style guarantee for weak convergence. It improves on Corollary~\ref{cor:s-flip-convergence} because it only requires the nuisance estimators converge at a rate of $n^{-1/2}$ in product \emph{at each timepoint}. There is no dependence across timepoints, unlike in Corollary~\ref{cor:s-flip-convergence}.

\section{Illustrative data analysis} \label{sec:data-analysis}

In this section, we provide illustrative results from a data analysis examining the effect of union membership on wages. Appendix~\ref{app:simulations} contains a simulation study. Our code is available at \url{https://github.com/alecmcclean/flip-interventions}. It uses \texttt{flip}, a development branch of the \texttt{lmtp} package \citep{williams2023lmtp}, available at \url{https://github.com/alecmcclean/lmtp}.

\bigskip

\subsection{Data Description}

We use the \texttt{wagepan} dataset from the \texttt{wooldridge} package in \texttt{R} \citep{wooldridge2024data, r2024language}. This dataset is from \citet{vella1998whose} and was obtained from the \textit{Journal of Applied Econometrics} archive at \url{http://qed.econ.queensu.ca/jae/}. The dataset contains employment information on 545 workers over eight years, from 1980-1987, though we focus on the first four years (1980-1983) for our analysis.

\bigskip

The dataset includes baseline covariates measured for each person: years of education, race (black/white), and ethnicity (hispanic/not hispanic). Additionally, several time-varying covariates are recorded for each year, including marital status (married/not married), health status (poor health: yes/no), labor market experience (in years), number of hours worked, occupation and industry classifications, region of residence (South/non-South), union membership, and the natural logarithm of hourly wage. An individual identifier allows us to link observations over time for each worker.

\subsection{Methodology}

Following \citet{vella1998whose}, we treat union membership as a time-varying treatment variable. We examine two flip interventions: one targeting \emph{always treated} ($a_t\equiv1$, i.e., always in a union) and another targeting \emph{never treated} ($a_t\equiv0$, i.e., never in a union). We employ the following flipping weight:
\[
f_t \{ H_t(\overline D_{t-1}) \} = 1 - \exp \Big[ - 20 \cdot \bbP \{ A_t(\overline D_{t-1}) = a_t \mid H_{t-1}(\overline D_{t-1}) \} \Big].
\]
This weight serves as a smooth approximation of a trimming indicator using the target propensity score. Our primary outcome of interest is the log wage in 1983. We estimate the mean difference in potential outcomes under each flip intervention and the absolute change in the number of treatments at each timepoint, then combine this information to estimate the longitudinal interventional flip effect between these two treatment regimes.

\bigskip

To estimate these effects, we use the sequentially doubly robust-style estimator in Algorithm~\ref{alg:seq-dr-est} with five-fold cross-fitting. For nuisance function estimation, we employ the SuperLearner, which combines multiple machine learning approaches: a linear model, generalized linear model, lasso with no interactions, a regression tree, and a random forest with default settings \citep{polley2024super, wright2017fast, therneau2023rpart}.

\subsection{Results}

Figure~\ref{fig:prop-scores} displays the estimated propensity scores through boxplots showing the estimated probability of union membership in each year. Notably, some probabilities are near zero, indicating that certain workers had close to zero probability of being in a union given their observed history. This represents a reasonable scenario for considering flip interventions that adapt to positivity violations. Figure~\ref{fig:trt-diff} illustrates the distribution of the mean difference in treatments under the always treated flip intervention minus the never treated flip intervention, showing roughly 0.9 treatments on average at each timepoint. This demonstrates that the two interventions quite closely approximate the relevant static interventions, although near positivity violations mean that some subjects are not always shifted to the target treatment under each intervention.

\begin{figure}[ht]
    \centering
    \includegraphics[width=0.8\linewidth]{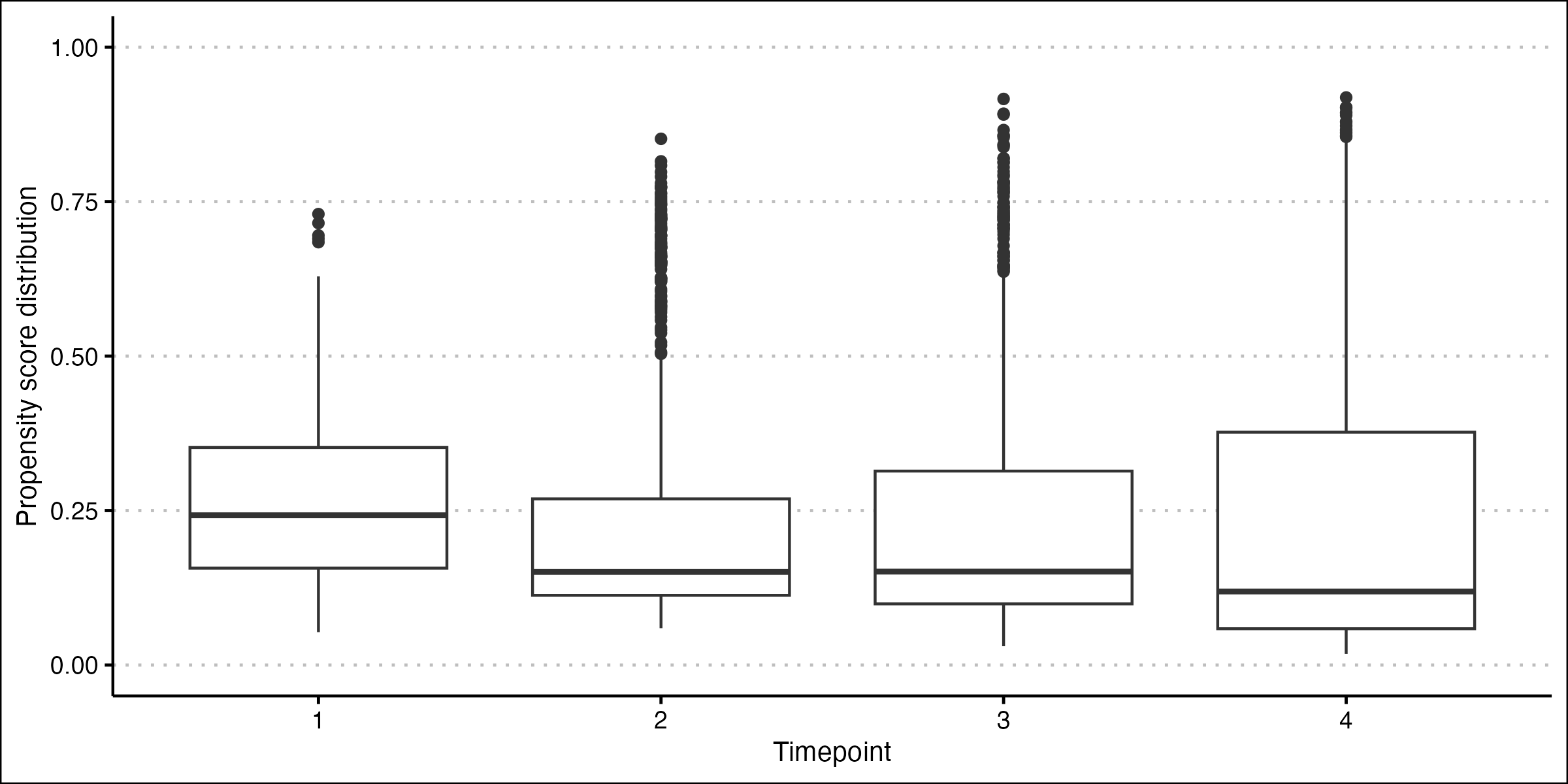}
    \caption{Propensity score distributions by timepoint.}
    \label{fig:prop-scores}
\end{figure}

\begin{figure}[ht]
    \centering
    \includegraphics[width=0.7\linewidth]{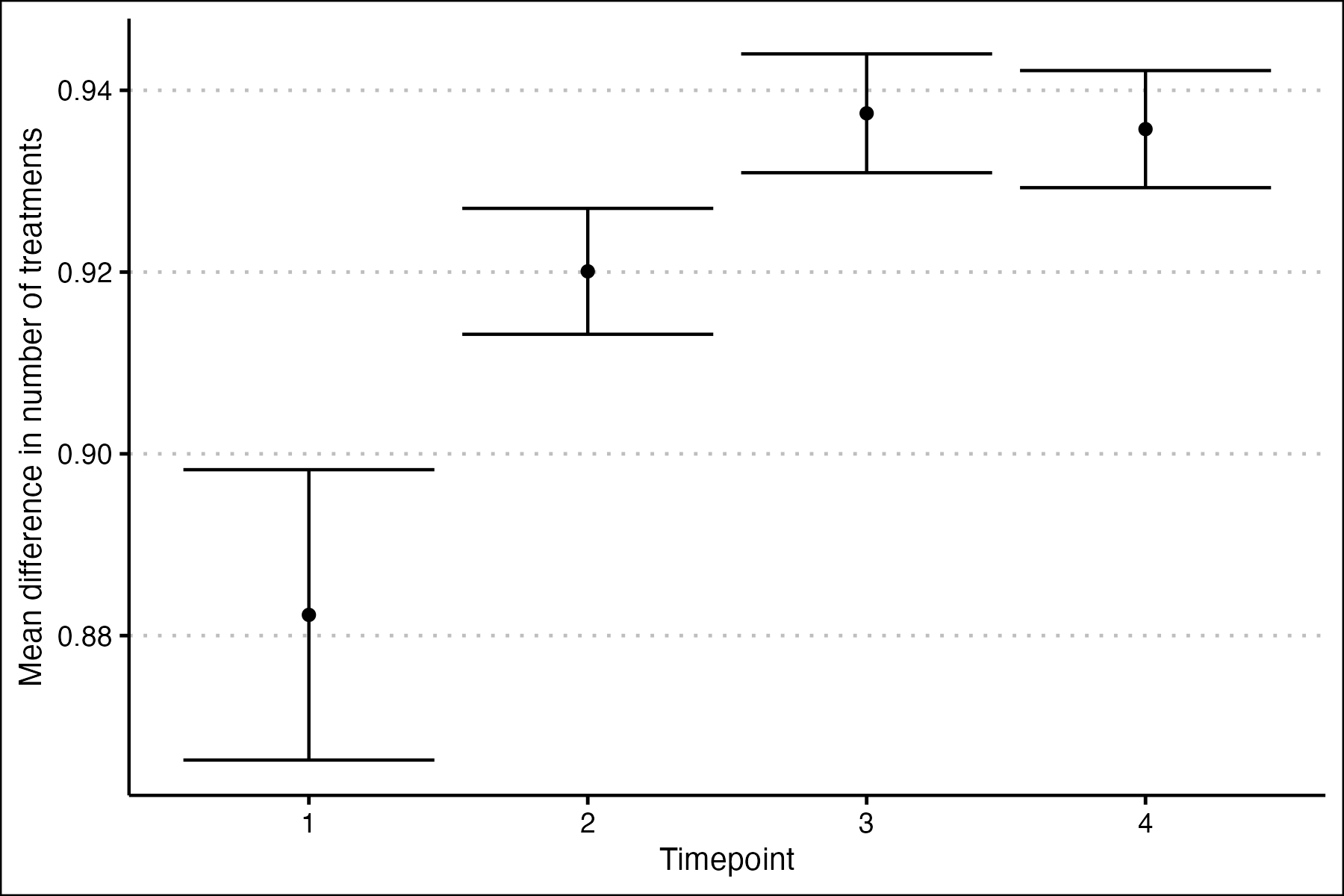}
    \caption{Mean difference in number of treatments under the ``always treat'' flip interventions and ``never treat'' flip interventions by timepoint.}
    \label{fig:trt-diff}
\end{figure}

\bigskip
Table~\ref{tab:results} presents our main findings. The average log wage was higher under the always treated flip intervention than the never treated flip intervention: 0.059 (95\% CI: [0.019, 0.098]). This corresponds to a roughly 6\% wage increase in 1983. Under the always treated flip intervention, roughly 92\% more workers (95\% CI: [0.911, 0.927]) would belong to a union compared to the never treated flip intervention, averaged across timepoints. When standardized by the increase in union membership, the estimated wage increase remains substantial at roughly 6.6\% per worker shifted into union membership per timepoint (0.064, 95\% CI: [0.021, 0.107]). These estimates are statistically significant at the $\alpha = 0.05$ level, as the 95\% confidence intervals for both the mean difference in potential outcomes and the longitudinal interventional flip effect exclude zero. These results suggest that union membership had a causally positive effect on wages in the early 1980's. However, a full analysis would need to incorporate sensitivity analyses to violations of the unmeasured confounding assumption, which may not be justified in this setting.

\begin{table}[H]
    \centering
    \begin{tabular}{lcc}
    \toprule
    \textbf{Parameter} & \textbf{Estimate} & \textbf{95\% confidence interval} \\ 
    \midrule
    Mean difference in potential outcomes & 0.059 & [0.019, 0.098] \\[4pt]
    Per-timepoint average change in treatments & 0.919 & [0.911, 0.927] \\[4pt]
    Longitudinal interventional flip effect & 0.064 & [0.021,\,0.107] \\
    \bottomrule
    \end{tabular}
    \caption{Results from wage panel data analysis.}
    \label{tab:results}
\end{table}

\section{Discussion} \label{sec:discussion}

In this paper, we proposed flip interventions as a novel approach for weighting and trimming in longitudinal data.  Flip interventions target a specific regime: if the subject would have taken the target treatment at a given timepoint, then there is no intervention; otherwise, the subject is flipped towards the target treatment with probability according to their weight (e.g., overlap weight, trimming indicator). We demonstrated that these interventions are identifiable under arbitrary positivity violations and are, in principle, practically implementable as single-world interventions. Moreover, in single-timepoint data we demonstrated that interventional flip effects, which standardize the mean difference in potential outcomes by the mean difference in treatments, are equivalent to weighted average treatment effects. Building on this framework, we introduced and analyzed longitudinal interventional flip effects. Finally, we developed efficient estimation strategies, establishing multiply robust and sequentially doubly robust estimators with favorable convergence guarantees.

\bigskip

While flip interventions provide a practical and robust method for addressing positivity violations in longitudinal settings, there are several promising directions for future work. One important area involves developing parsimonious summaries of effects. Although the benefit of flip interventions is they allow researchers to directly target each treatment regime while adapting to positivity violations, it may still be desirable to obtain parsimonious summaries of effect estimates, rather than target each of the $2^T$ regimes separately. Marginal structural models could be combined with flip effects to do so.

\bigskip

Another compelling future direction is motivated by an alternative typical perspective on trimming and weighting methods, which emphasizes minimizing estimation error, rather than robustness to positivity violations. Therefore, an alternative approach could explicitly examine estimation variance under different weights, targeting data-adaptive parameters optimized for minimal variance.  Implementation could leverage cross-validation or other model selection strategies, though the development of asymptotic theory for fully data-driven selection of tuning parameters remains an open challenge and promising direction for future research \citep{khan2022doubly}.

\section*{References}
\vspace{-0.3in}
\bibliographystyle{plainnat}
\bibliography{references}

\newpage
\appendix

\newgeometry{margin=1in} 
\section*{Appendix}

This appendix does the following:
\begin{itemize}
    \item[Appendix~\ref{app:trt_prob}] provides efficient estimators for the average treatment value.
    \item[Appendix~\ref{app:simulations}] provides simulation results reinforcing our theoretical results in the main paper.
    \item[Appendix~\ref{app:t=1}] provides proofs for the results in Section~\ref{sec:t=1}.
    \item[Appendix~\ref{app:ltt}] provides proofs for the results in Section~\ref{sec:ltt}.
    \item[Appendix~\ref{app:est}] provides proofs for the results in Section~\ref{sec:efficiency}. 
    \item[Appendix~\ref{app:additional}] provides proofs for the results in Appendix~\ref{app:trt_prob}.
\end{itemize}

\section{Estimating the average treatment value} \label{app:trt_prob}

For longitudinal interventional flip interventions defined in \eqref{eq:long-int-flip-effect}, we must also estimate $\bbE\{ D_{f_t}(a_t) \}$ and $\bbE\{ D_{f_t}(a_t^\prime) \}$. These estimands are identified in Theorem~\ref{thm:id-flip}. In this section, we focus on estimating $\bbE\{ D_{f_T}(a_T) \}$ when $\{D_{f_1}(a_1), \dots, D_{f_T}(a_T) \}$ are flip interventions with smooth weights, as in Section~\ref{sec:efficiency}. Adapting the estimator for $t < T$ is straightforward. Using the notation from the main paper, the identified estimand is
\begin{equation} \label{eq:trt_est}
    \psi_D = \bbE \left[ Q_T(A_T =1 \mid H_T) \left\{ \prod_{t=1}^{T-1} r_t(A_t \mid H_t) \right\}\right]
\end{equation}
where $r_t(A_t \mid H_t) = \frac{Q_t(A_t \mid H_t)}{\bbP(A_t \mid H_t)}$, 
\[
Q_t(A_t = b_t \mid H_t) = \bbP(A_T = b_T \mid H_T) + \left\{ 2 \one(b_t = a_t) - 1 \right\} s_t \{ \bbP(A_t = a_t \mid H_t) \} \{ 1 - \bbP(A_t = a_t \mid H_t) \},
\]
and $s_t$ is a smooth and twice differentiable function of the target propensity score with bounded derivatives. 

\bigskip

Next, we define the sequential regression as $m_T(0, H_T) = 0$ and $m_T(1, H_T) = 1$ and then recursively define
\[
m_t(A_t, H_t) = \bbE \left\{ \sum_{b_{t+1}} m_{t+1}(b_{t+1}, H_{t+1}) Q_{t+1}(b_{t+1} \mid H_{t+1}) \mid A_t, H_t \right\}.
\]
for $t < T$, where $m_0 = \psi_D$. Then, we can derive the efficient influence function.
\begin{proposition} \label{prop:trt_prob_eif}
    Let $\psi_D$ be as in \eqref{eq:trt_est}. Under the setup of Proposition~\ref{prop:eif} the  efficient influence function of $\psi_D$ under a nonparametric model is
    \begin{align*}
        \varphi_D(Z) = \sum_{t=1}^{T} \left\{ \prod_{s=1}^{t-1} r_s(A_s \mid H_s) \right\} \left[ \sum_{b_t} m_t(b_t, H_t) \left\{ Q_t(b_t \mid H_t) + \phi_t(b_t; A_t, H_t) \right\} - m_{t-1}(A_{t-1}, H_{t-1}) \right].
    \end{align*}
\end{proposition}

\begin{proof}
    The proof is deferred to Appendix~\ref{app:additional}.
\end{proof}

One can construct a multiply robust-style estimator based on this efficient influence function. It is a straightforward adaptation of Algorithm~\ref{alg:mult-dr-est} from the main paper.  As in the main paper, we'll focus on sample splitting and cross-fitting with two splits, and without cycling the folds.

\begin{algorithm}[Multiply robust-style estimator] \label{alg:mult-dr-est-trt-prob} Assume training and evaluation dataset of $n$ observations.
    \begin{enumerate}
        \item For all timepoints, regress $A_t$ on $H_t$ in the training data and obtain propensity score models; using these models, compute the intervention propensity scores $\widehat Q(A_t \mid H_t)$, ratios $\widehat r_{t}(A_t \mid H_t)$, and efficient influence functions $\widehat \phi_{t}(b_t; A_t, H_t)$ for all samples and timepoints.
        \item For $t=T-1$ to $t=1$: 
        \begin{enumerate}
            \item 
            \begin{enumerate}
                \item If $t = T-1$, then $\widehat P_{T}(H_{T}) = \widehat Q_T(A_T = 1 \mid H_T)$. Otherwise, pseudo-outcome $\widehat P_{t+1}(H_{t+1})$ is available from the previous step in loop (see step 3 below).
                \item In the training data, regress $\widehat P_{t+1}(H_{t+1})$ against $A_t$ and $H_t$ to obtain a sequential regression model; using this model, across the full data, obtain estimates $\widehat m_{t}(0, H_t), \widehat m_{t}(1, H_t)$.
            \end{enumerate}
        \item Across the full data, compute pseudo-outcomes \\ $\widehat P_t(H_t) = \widehat m_t(0, H_t) \widehat Q_t(0 \mid H_t) + \widehat m_t(1, H_t) \widehat Q_t(1 \mid H_t)$ to use in the next step.
        \end{enumerate} 
        \item In the evaluation data only, compute the plug-in estimator $\widehat m_0 = \bbP_n \{ \widehat P_1(X_1) \}$ and efficient influence function values by plugging nuisance estimates in to $\varphi_D(Z)$ in Proposition~\ref{prop:trt_prob_eif}.
    \end{enumerate}
    Finally, output the point estimate and variance estimate 
    \[
    \widehat \psi_D := \widehat m_0 + \bbP_n \{ \widehat \varphi_D(Z) \} \text{ and } \widehat \sigma_D^2 := \bbP_n \big\{ \widehat \varphi_D(Z)^2 \big\}.
    \]
\end{algorithm}
\noindent We have the following bound on the bias of the estimator from Algorithm~\ref{alg:mult-dr-est-trt-prob}.
\begin{theorem} \label{thm:trt-bias}
    Under the setup of Theorem~\ref{thm:flip-bias}, let $\psi_D$ be as in \eqref{eq:trt_est} and $\widehat \psi_D$ denote the estimator from Algorithm~\ref{alg:mult-dr-est-trt-prob}. Moreover, let $\widetilde m_t(A_t, H_t) = \bbE \left\{ \sum_{b_{t+1}} \widehat m_{t+1}(b_{t+1}, H_{t+1}) \widehat Q_{t+1}(b_{t+1} \mid H_{t+1}) \mid A_t, H_t \right\}$. Then,
    \begin{align*}
        \left| \bbE \left( \widehat \psi_D - \psi_D \right) \right| \lesssim \min &\bigg\{ \sum_{t=1}^{T} \| \widehat m_t - m_t \| \| \widehat \pi_t - \pi_t \| + \| \widehat \pi_t - \pi_t \|^2, \\
        &\sum_{t=1}^{T} \| \widehat m_t - \widetilde m_t \| \left( \sum_{s=1}^{t} \| \widehat \pi_s - \pi_s \| \right) + \| \widehat \pi_t - \pi_t \| \left( \sum_{s=1}^{t} \| \widehat \pi_s - \pi_s \| \right).
    \end{align*}
\end{theorem}

\begin{proof}
    This follows by the same analysis as for the efficient estimator of $\bbE \{ Y(\overline D_T) \}$ and using the algebra in the proof of Proposition~\ref{prop:trt_prob_eif}.  We omit it for brevity.
\end{proof}
\noindent Then, we have the following convergence guarantee.
\begin{corollary}
    Under the setup of Theorem~\ref{thm:trt-bias}, let $\widehat \sigma_D^2$ be a variance estimate from Algorithm~\ref{alg:mult-dr-est-trt-prob}.  Suppose $\| \widehat \varphi_D - \varphi_D \| = o_\bbP(1)$ and 
    \begin{align*}
        \min &\bigg\{ \sum_{t=1}^{T} \| \widehat m_t - m_t \| \| \widehat \pi_t - \pi_t \| + \| \widehat \pi_t - \pi_t \|^2, \\
        &\sum_{t=1}^{T} \| \widehat m_t - \widetilde m_t \| \left( \sum_{s=1}^{t} \| \widehat \pi_s - \pi_s \| \right) + \| \widehat \pi_t - \pi_t \| \left( \sum_{s=1}^{t} \| \widehat \pi_s - \pi_s \| \right) = o_\bbP(n^{-1/2}).
    \end{align*}
    Then,
    \[
    \sqrt{\frac{n}{\widehat \sigma_D^2}} (\widehat \psi_D - \psi_D) \indist N(0,1).
    \]
\end{corollary}
As in the main paper with the following result, we motivate the sequentially doubly robust-style estimator with a result on the debiased pseudo-outcome.
\begin{lemma} \label{lem:seq-property-trtment-prob}
    Under the setup of Proposition~\ref{prop:trt_prob_eif}, define $P_T^\ast(Z) = Q_T(A_t = 1 \mid H_T) + \phi_T(1; A_T, H_T)$ and recursively define for $t=T-1$ to $t=1$ \footnotesize
    \begin{align*}
        &P_t^\ast(Z) = \sum_{b_t} m_t (b_t, H_t) \Big\{ Q_t(b_t \mid H_t) + \phi_t(b_t; A_t, H_t) \Big\} \\
        &+ \sum_{s=t}^{T-1} \left\{ \prod_{k=t}^s r_k(A_k \mid H_k) \right\} \left\{ \sum_{b_{s+1}} m_{s+1} (b_{s+1}, H_{s+1}) \Big\{ Q_{s+1}(b_{s+1} \mid H_{s+1}) + \phi_{s+1}(b_{s+1}; A_{s+1}, H_{s+1}) \Big\} - m_s(A_s, H_s) \right\}.
    \end{align*}
    \normalsize Then,
    \begin{equation} \label{eq:pseudo-unbiased-2}
        \bbE \left\{ P_{t+1}^\ast(Z) \mid A_t, H_t \right\} = m_t(A_t, H_t).
    \end{equation}
    Moreover, suppose access to fixed nuisance estimates $\left\{ \widehat m_s^\ast, \widehat Q_s \right\}_{s=t+1}^{T}$ to construct $\widehat P_{t+1}^\ast(Z)$. Then, 
    \begin{align}
        &\bbE \left\{ \widehat P_{t+1}^\ast (Z) - m_t(A_t, H_t) \mid A_t, H_t \right\} = \nonumber \\
        &\hspace{0.2in}\sum_{s=t+1}^{T} \bbE\left[ \left\{ \prod_{k=t+1}^{s-1} \widehat r_k(A_k \mid H_k) \right\} \Big\{ m_s(A_s, H_s) - \widehat m_s^\ast(A_s, H_s) \Big\} \Big\{ \widehat r_s(A_s \mid H_s) - r_s (A_s \mid H_s) \Big\} \mid A_t, H_t \right] \nonumber \\
        &+ \sum_{s=t+1}^{T} \bbE \left[ \left\{ \prod_{k=t+1}^{s-1} \widehat r_k(A_k \mid H_k) \right\} \sum_{b_s} \widehat m_s^\ast(A_s, H_s) \left\{ \widehat Q_s(b_s \mid H_s) + \widehat \phi_s(b_s; A_s, H_s) - Q_s(b_s \mid H_s) \right\} \mid A_t, H_t \right]. \label{eq:pseudo-dr-2}
    \end{align}
\end{lemma}
\begin{proof}
    See Appendix~\ref{app:additional}.
\end{proof}
\noindent This result suggests the simple amendment to the estimator from Algorithm~\ref{alg:mult-dr-est-trt-prob}.
\begin{algorithm}[Sequentially doubly robust-style estimator] \label{alg:seq-dr-est-trt-prob}
    Use Algorithm~\ref{alg:mult-dr-est-trt-prob} with the following amendments to the sequential regression loop:
    \begin{itemize}[itemsep=0.03in]
        \item In step 2(a)(i), let $\widehat P_{T}^\ast(Z) = \widehat Q_T(A_T =1 \mid H_T) + \widehat \phi_T(1; A_T, H_T)$.
        \item In step 2(a)(ii), in the training data regress $\widehat P_{t+1}^\ast(Z)$ against $A_t$ and $H_t$ and, in the evaluation data, obtain estimates $\widehat m_{t}^\ast (0, H_t), \widehat m_{t}^\ast (1, H_t)$.
        \item In step 3, when constructing pseudo-outcomes, use the transformation $\widehat P_t^\ast(Z)$ which uses available nuisance estimates $\left\{ \widehat m_s^\ast, \widehat Q_s \right\}_{s=t+1}^{T}$.
    \end{itemize}
    Finally, construct a point estimate and variance estimate as
    \[
    \widehat \psi_D^\ast = \bbP_n \{ \widehat P_1^\ast(Z) \} \text{ and } \widehat \sigma_D^2 = \bbP_n \Big[ \big\{ \widehat P_1^\ast(Z) - \widehat \psi_D^\ast \big\}^2 \Big].
    \]
\end{algorithm}
\noindent Finally, we have the following bias bound and convergence guarantee.
\begin{theorem} \label{thm:seq-dr-bias-trt-prob}
    Under the setup of Theorem~\ref{thm:trt-bias}, let $\widehat \psi_D^\ast$ denote a point estimate from Algorithm~\ref{alg:seq-dr-est-trt-prob} and let $\widetilde m_t^\ast (A_t, H_t) = \bbE \left\{ \widehat P_{t+1}^\ast(Z) \mid A_t, H_t \right\}$. Then,
    \begin{align*}
        \left| \bbE \left( \widehat \psi_D^\ast - \psi_D \right) \right| \lesssim \sum_{t=1}^{T} \| \widehat \pi_t - \pi_t \| \Big( \| \widehat m_t^\ast - \widetilde m_t^\ast \| + \| \widehat \pi_t - \pi_t \| \Big).
    \end{align*}
\end{theorem}

\section{Simulation study} \label{app:simulations}

Our simulation study leveraged a simple data generating process with two timepoints:
\begin{align*}
    X_1 &\sim \text{Unif}(0,1), \\
    A_1 &\sim \text{Bern}\{ \bbP(A_1 \mid X_1) \} \text{ where } \bbP(A_1 \mid X_1) = \one (0.1 \leq X_1 \leq 0.9) \left( \tfrac{X_1 - 0.1}{0.8} \right), \\
    X_2(A_1, X_1) &= \tfrac{X_1 + A_1}{2}, \\
    A_2 &\sim \text{Bern} \{ \bbP(A_2 \mid H_2) \} \text{ where } \bbP(A_2 \mid H_2) = \one (0.1 \leq X_2 \leq 0.9) \left( \tfrac{X_2 - 0.1}{0.8} \right)\text{, and }\\
    Y(A_2, A_1) &= X_1 + X_2 + A_1 + A_2 + N(0,1).
\end{align*}
We constructed nuisance estimators by adding random error to the true nuisance functions, with the magnitude of the random error decreasing as sample size increased. This design emulates nuisance estimators whose root mean squared error converges at specified rates with increasing sample size, enabling us to examine how the sequentially doubly robust-style estimator performed under different convergence scenarios. The four nuisance estimators we considered were the two propensity score estimators and the two sequential regression estimators. We structured our simulations such that both propensity score estimators converged at the same rate, and both sequential regressions converged at the same rate. We conducted 250 simulation runs for each combination of sample size and convergence rates.

\bigskip

Our focus was on estimating the mean difference in potential outcomes under flip interventions that targeted always-treated and never-treated regimes. We considered two distinct weighting schemes: overlap weights and smooth trimming weights, with the smooth trimming indicator defined as $1 - \exp(-10x)$.

\bigskip

Figure~\ref{fig:simulation-results} summarizes our simulation results. The x-axis represents sample size, while the y-axis indicates the coverage probability of the sequentially doubly robust-style estimator, calculated over 250 simulation runs with 95\% confidence intervals. Columns correspond to varying estimation rates for sequential regressions, ranging from $n^{-0.1}$ to $n^{-0.5}$, while rows correspond to different estimation rates for propensity score estimators, also ranging from $n^{-0.1}$ to $n^{-0.5}$. The green points and associated error bars denote results for overlap weights, while orange points and error bars represent results for smooth trimming weights.

\bigskip

The simulation results corroborate our theoretical analysis in the main paper. Consistent with Corollary~\ref{cor:s-flip-seq-convergence}, accurate estimation of propensity scores at least at the $n^{-1/4}$ rate is necessary for the sequentially doubly robust-style estimator to achieve $\sqrt{n}$-consistency and Gaussian convergence to the true estimand. This requirement is clearly demonstrated by the top row, where estimators fail to converge due to inadequate propensity score estimation (data points with coverage less than 0.5 were omitted from the figure). When propensity scores are sufficiently well estimated (second row), convergence properties depend on the sequential regression rates. Specifically, when the product of estimation rates is slower than $n^{-1/2}$ (left-most column, second row), estimators do not achieve correct coverage. However, in all other cases within the second row, and throughout the entire bottom row, the estimators exhibit appropriate coverage and consistently converge to the true estimand as sample size increases.

\begin{figure}
    \centering
    \includegraphics[width=\linewidth]{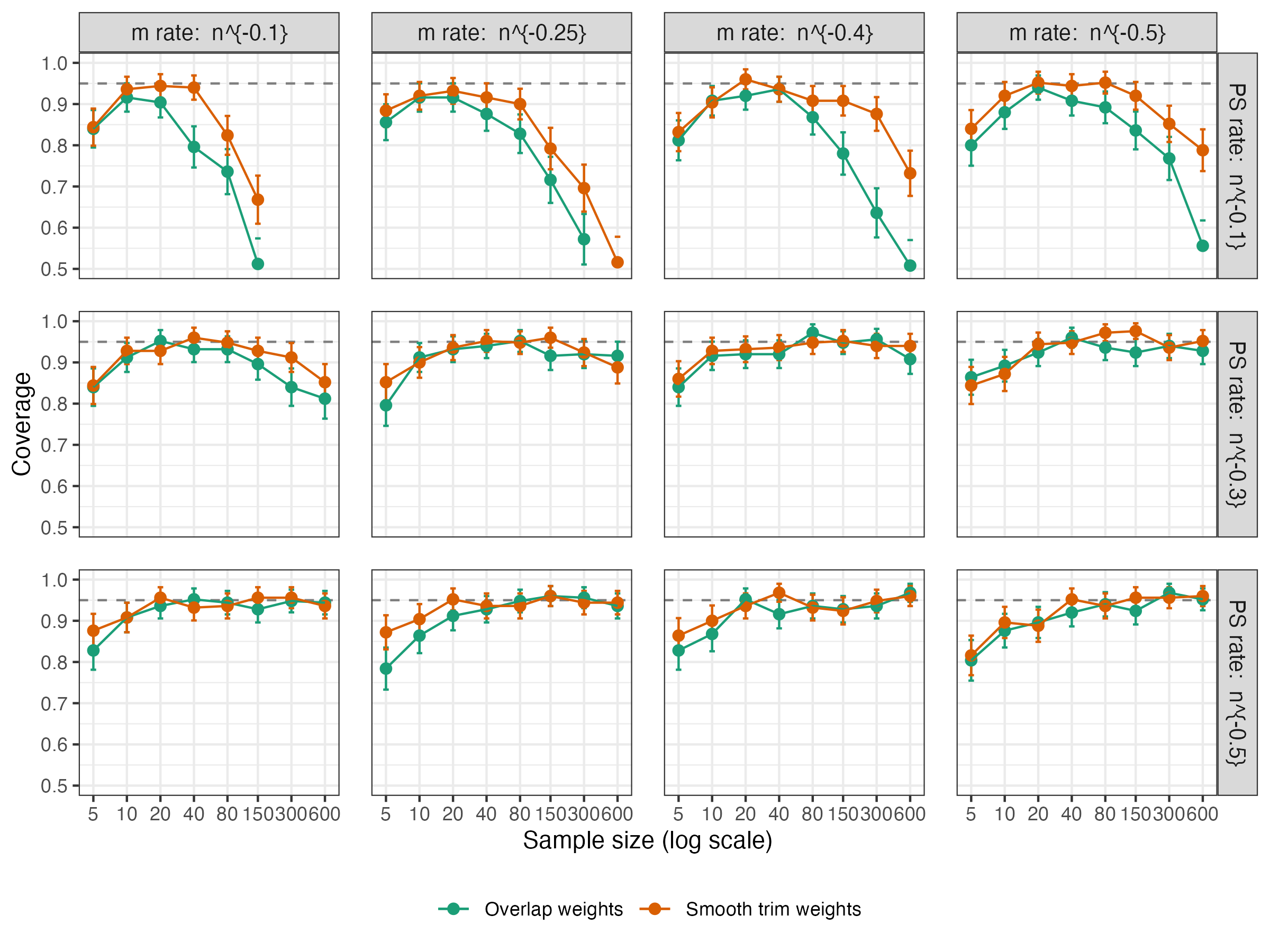}
    \caption{Simulations results}
    \label{fig:simulation-results}
\end{figure}

\section{Proofs for Section~\ref{sec:t=1}: Proposition~\ref{prop:equivalence}} \label{app:t=1}

\begin{proof}
Since $(D_f(0), D_f(1)) = \begin{cases}
    (A, A), & \text{ if } V > f(X)\\
    (0,1), &\text{ if } V \leq f(X)
\end{cases}$, by consistency the numerator satisfies
\begin{align*}
    \bbE \big[ Y\{ D_f(1) \} - Y\{ D_f(0) \} \big] 
    &= \bbE \big[ \one\{V \leq f(X)\}\{Y(1) - Y(0)\}\big] \\
    &= \bbE(\bbE \big[ \one\{V \leq f(X)\}\{Y(1) - Y(0)\} \mid X\big]) \\
    &= \bbE(\bbE \big[ \one\{V \leq f(X)\} \mid X\big] \cdot \bbE \big[Y(1) - Y(0) \mid X\big]) \\
    &= \bbE(f(X)\bbE\big[Y(1) - Y(0) \mid X\big]),
\end{align*}
where the third line results from $V \ind (Y(0), Y(1)) \mid X$. Similarly,
\[\bbE \big[D_f(1)- D_f(0) \big] = \mathbb{P}[V \leq f(X)] = \bbE[\mathbb{P}[V \leq f(X) \mid X] ] = \mathbb{E}[f(X)].\]
This also yields the conditional interpretation:
\[\bbE[Y(1) - Y(0) \mid V \leq f(X)] = \frac{\bbE[\one\{V \leq f(X)\}\{Y(1) - Y(0)\}]}{\mathbb{P}[V \leq f(X)]} = \frac{\bbE \big[ Y\{ D_f(1) \} - Y\{ D_f(0) \} \big] }{\bbE \big[D_f(1)- D_f(0) \big]}.\]
\end{proof}

\begin{remark}
    The proof of Proposition~\ref{prop:equivalence} still holds without requiring that the two interventions depend on the same auxiliary random variable $V$. Consider the case when \( D_f(0) = A \one \{ V_0 > f(X) \} \) and \( D_f(1) = A + (1-A) \one\{ V_1 \leq f(X) \}\) and $V_0 \ind V_1$. We then have
    \[
    (D_f(0), D_f(1)) = \begin{cases}
        (A, A) & \text{ if } V_0 > f(X) \text{ and } V_1 > f(X), \\
        (0,A) &\text{ if } V_1 \leq f(X) < V_0, \\
        (A,1) &\text{ if } V_0 \leq f(X) < V_1\text{, and } \\
        (0,1) &\text{ if } V_0 \leq f(X) \text{ and } V_1 \leq f(X).
    \end{cases}
    \]
    Therefore, by consistency,
    \begin{align*}
        \bbE \big[ Y\{ D_f(1) \} - Y\{ D_f(0) \} \big] 
        &= \bbE \big[ \one\{V_0 \leq f(X)\} \one\{ V_1 \leq f(X) \} \{Y(1) - Y(0)\}\big] \\
        &\hspace{0.2in}+ \bbE \big[ \one\{ V_1 \leq f(X) \} \one \{ V_0 > f(X) \} \{ Y(A) - Y(0) \} \big] \\
        &\hspace{0.2in}+ \bbE \big[ \one \{ V_0 \leq f(X) \} \one \{ V_1 > f(X) \} \{ Y(1) - Y(A) \} \big] \\
        &= \bbE \big[ f(X)^2 \bbE \{ Y(1) - Y(0) \mid X \} \big] \\
        &\hspace{0.2in}+ \bbE \big[ f(X) \{ 1 - f(X) \} \bbE \{ Y(A) - Y(0) \mid X \} \big] \\
        &\hspace{0.2in}+ \bbE \big[ f(X) \{ 1 - f(X) \} \bbE \{ Y(1) - Y(A) \mid X \} \big] \\
        &= \bbE \big[ f(X)^2 \bbE \{ Y(1) - Y(0) \mid X \} \big] + \bbE \big[ f(X) \{ 1 - f(X) \} \bbE \{ Y(1) - Y(0) \mid X \} \big] \\
        &= \bbE \big[ \bbE \{ Y(1) - Y(0) \mid X \} f(X) \big].
    \end{align*}
    The second equality follows by iterated expectations on $X$ and because $V_0 \ind V_1$ and $(V_0, V_1) \ind \{ Y(0), Y(1) \}$ by assumption.  
\end{remark}

\section{Proofs for Section~\ref{sec:ltt}} \label{app:ltt}

\subsection{Helper results}

For the identification results, we provide several helper lemmas.  We also slightly amend our notation from the main paper, so we can specify the dependence of random variables on the counterfactual past interventions. In what follows, we make the following ``assumption'' (which corresponds to the setup in the main paper).
\begin{assumption} \label{asmp:setup}
    Suppose the following setup:
    \begin{itemize}
        \item $\{ D_1, D_2(D_1), \dots, D_T(\overline D_{T-1}) \}$ denote a set of treatment decisions where \\ $D_t(\overline D_{t-1}) = d_t \{ A_t(\overline D_{t-1}), H_t(\overline D_{t-1}), V_t \}$ for some deterministic function $g_t$, where $V_1, \dots, V_T$ are mutually independent and $V_t \ind Z$ for all $t \in \{1, \dots, T\}$, 
        \item $\overline X_t(\overline a_{t-1})$ denotes the natural covariate history under an intervention that sets treatment to $\overline a_{t-1}$ up until time $t-1$, 
        \item $H_t(\overline a_{t-1}) = \big\{ \overline X_{t-1}(\overline a_{t-1}), \overline D_{t-1} = \overline a_{t-1} \big\}$ denotes the natural covariate history and the intervention treatment history,
        \item $A_t(\overline a_{t-1})$ denotes the natural value of treatment after history $H_t(\overline a_{t-1})$, 
        \item $D_t(\overline a_{t-1}) = d_t \{ A_t(\overline a_{t-1}), H_t(\overline a_{t-1}), V_t \}$, i.e., the treatment decision at time $t$ is random via $A_t(\overline a_{t-1}), H_t(\overline a_{t-1})$, and $V_t$, 
        \item $Y( \overline a_t, \underline D_{t+1} )$ denotes the potential outcome under an intervention that sets treatment $\overline a_t$ up until time $t$ and assigns treatment according to treatment decisions \\ $\underline D_{t+1} = \{D_{t+1}(\overline a_{t}), D_{t+2}(D_{t+1}, \overline a_t), \dots D_T (\overline a_t, D_{t+1}, \dots, D_{T-1}) \}$ thereafter.
    \end{itemize}
\end{assumption}    

The first result gives us the important consistency steps we use within the g-formula. In essence, it says that if we condition on an observed history up until the end of timepoint $t-1$, then the counterfactual history in timepoint $t$ is equal to the observed history. An important corollary is that, conditional on an observed history up until the end of timepoint $t-1$, the propensity scores at timepoint $t$ are identified.
\begin{proposition} \label{prop:consistency}
    Conditional on $\{ \overline X_{t-1}, \overline A_{t-1} = \overline b_{t-1} \}$, 
    \begin{itemize}
        \item $H_t(\overline b_{t-1}) = \big\{ \overline X_t, \overline A_{t-1} = \overline b_{t-1} \big\}$ and
        \item $A_t(\overline b_{t-1}) = A_t$.
    \end{itemize}
    As a consequence, conditional on $\{ \overline X_{t-1}, \overline A_{t-1} = \overline b_{t-1} \}$,
    \[
    \bbP \{ A_t(\overline b_{t-1}) = b_t \mid H_t(\overline b_{t-1}) \} = \bbP(A_t = b_t \mid \overline X_t, \overline A_{t-1} = \overline b_{t-1}).
    \]
    In other words, the propensity score at timepoint $t$ is identified.
\end{proposition}

\begin{proof}
    These follow by the consistency assumption in the NPSEM. 
\end{proof}

The next result gives us the exchangeability necessary for identification when the intervention depends on the natural value of treatment.
\begin{lemma} \label{lem:exch-short}
    Under Assumption~\ref{asmp:strong-seq-exch} and Assumption~\ref{asmp:setup}, 
    \begin{align*}
        A_t(\overline a_{t-1}) &\ind Y (\overline a_t, \underline D_{t+1}) \mid H_t(\overline a_{t-1}) \text{ and } \\
        D_t(\overline a_{t-1}) &\ind Y (\overline a_t, \underline D_{t+1}) \mid H_t(\overline a_{t-1}).
    \end{align*}
\end{lemma}

\begin{proof}
    Conditional on $H_t(\overline a_{t-1})$, $A_t(\overline a_{t-1})$ only depends on the random variable $U_{A,t}$. Meanwhile, $Y(\overline a_t, \underline D_{t+1})$ depends on $(\underline V_{t+1}, \underline U_{A,t+1}, \underline U_{X,t+1}, U_Y)$. Both results then follow by Assumption~\ref{asmp:strong-seq-exch} and the assumption on $\overline V_T$.
\end{proof}

The next result gives the same exchangeability result when the intervention does not depend on the natural value of treatment. It only requires standard sequential randomization.  
\begin{lemma} \label{lem:exch-short-no-strong}
    Under Assumption~\ref{asmp:setup}, suppose instead that
    \[
    D_t(\overline a_{t-1}) = d_t \{ \bbP \{ A_t(\overline a_{t-1}) \mid H_t(\overline a_{t-1}) \}, H_t(\overline a_{t-1}), V_t \}
    \]
    for some function $d_t$; i.e., the intervention only depends on the natural propensity score, the natural covariate history, and auxiliary randomness, but not the natural value of treatment. Then, under only Assumption~\ref{asmp:standard-seq-exch},
    \begin{align*}
        A_t(\overline a_{t-1}) &\ind Y (\overline a_t, \underline D_{t+1}) \mid H_t(\overline a_{t-1}) \text{ and } \\
        D_t(\overline a_{t-1}) &\ind Y (\overline a_t, \underline D_{t+1}) \mid H_t(\overline a_{t-1}).
    \end{align*}
\end{lemma}

\begin{proof}
    Conditional on the natural covariate history, the intervention is only random via $V_t$. Meanwhile, $\underline D_{t+1}$ are only random via the future natural covariate history and random variables $\underline V_{t+1}$. Therefore, conditional on $H_t(\overline a_{t-1})$, $Y(\overline a_t, \underline D_{t+1})$ is only random via $(\underline U_{X,t+1}, U_Y, \underline V_{t+1})$. Therefore, the result follows by Assumption~\ref{asmp:standard-seq-exch} and the assumption on $\overline V_T$.
\end{proof}

\begin{remark}
    The construction and result in Lemma~\ref{lem:exch-short-no-strong} applies to the stochastic amendment of flip interventions we suggested in Section~\ref{subsec:longitudinal-flip}:
    \begin{align*}
        D_{f_t}(a_t) = \one \Big( V_t \leq &\; \one(a_t = 1) f_t\{ H_t(\overline D_{t-1}) \} \\
        &+ \Big[ 1 - f_t\{ H_t(\overline D_{t-1}) \} \Big] \bbP\{ A_t(\overline D_{t-1}) =1 \mid H_t(\overline D_{t-1}) \} \Big) 
    \end{align*}
\end{remark}

Finally, we establish a result for positivity, which shows that as long as the weight satisfies condition~\ref{cond:positivity} of Theorem~\ref{thm:id-flip}, then flip interventions avoid positivity violations. 
\begin{lemma} \label{lem:absolute-continuity}
    Under the setup of Theorem~\ref{thm:id-flip}, let $a_t$ denote the target treatment and $Q_t$ denote the intervention propensity score. Then,
    \[
    \bbP \left\{ \bbP(A_t = 1-a_t \mid H_t) = 0 \implies Q_t(A_t = 1-a_t \mid H_t) = 0 \right\}.
    \]
    Moreover, if condition~\ref{cond:weight-construction} holds then 
    \[
    \bbP \big\{ \bbP(A_t = a_t \mid H_t) = 0 \implies Q_t(a_t \mid H_t) = 0 \big\} = 1.
    \]
\end{lemma}

\begin{proof}
    For the first result, notice that $\bbP(A_t = 1-a_t \mid H_t) = 0$ implies 
    \[
    Q_t(a_t \mid H_t) = 1 + 0 = 1,
    \]
    which implies $Q_t(1-a_t \mid H_t) = 0$.  

    \bigskip

    For the second result, condition~\ref{cond:weight-construction} asserts that $\bbP(A_t = a_t \mid H_t) = 0 \implies f_t(H_t) = 0$ almost always; therefore, $\bbP(A_t = a_t \mid H_t) = 0$ implies 
    \[
    Q_t(a_t \mid H_t) = 0 + 1 \cdot 0 = 0.
    \] 
\end{proof}

\subsection{Proof of Theorem~\ref{thm:id-flip}}

Finally, we have the full proof of the main theorem. For brevity, we use just $\overline D_T$ to denote the interventions.
\begin{proof}
    First, we have
    \begin{align*}
        \bbE \big\{ Y( \overline D_T ) \big\} &= \bbE \Big[ \bbE \big\{ Y( \overline D_T ) \mid X_1 \big\} \Big] = \bbE \Big[ \bbE \big\{ Y(\overline D_T) \mid D_1, X_1 \big\} \mid X_1 \Big] \\
        &= \bbE \Big[ \sum_{b_1} \bbE \big\{ Y( b_1, \underline D_2) \mid D_1 = b_1, X_1 \big\} \bbP ( D_1 = b_1 \mid X_1 ) \Big] \\
        &=\bbE \Big[ \sum_{b_1} \bbE \big\{ Y( b_1, \underline D_2 ) \mid A_1 = b_1, X_1 \big\} Q_1(b_1 \mid X_1) \Big] \\
        &\equiv \sum_{b_1} \int_{\mathcal{X}_1} \bbE \big\{ Y( b_1, \underline D_2 ) \mid A_1 = b_1, x_1 \big\} Q_1(b_1 \mid x_1) d\bbP(x_1) 
    \end{align*}
    where the first line follows by iterated expectations on $X_1$ and then on $X_1$ and $D_1$, the second by taking the expectation over $D_1$, the third by Lemma~\ref{lem:exch-short} in the inner expectation and by the definition of $D_1$ in the outer probability and Proposition~\ref{prop:consistency}, and the fourth by linearity of expectation and definition. 

    \bigskip

    That everything is well-defined follows by Lemma~\ref{lem:absolute-continuity}. First, by the construction of the interventions, when $b_1 = 1-a_1$ Lemma~\ref{lem:absolute-continuity} guarantees the outer expectation is well-defined.  The inner expectation might not be well-defined, but $Q_1 ( 1-a_1 \mid x_1 ) = 0$ whenever that occurs. Meanwhile, if condition~\ref{cond:weight-construction} holds, the same argument applies for $b_1 = a_1$ by Lemma~\ref{lem:absolute-continuity}. Finally, if condition~\ref{cond:weight-construction} does not hold but condition~\ref{cond:positivity} of the theorem holds, then the positivity assumption guarantees that $Q_1(a_1 \mid X_1)$ is almost never zero. 

    \bigskip

    \noindent The rest of the proof will follow by induction. We address the $t=2$ step. We have
    \begin{align*}
        \bbE \big\{ Y( b_1, \underline D_2 ) \mid A_1 = b_1, X_1 \big\} &= \bbE \Big[ \bbE \big\{ Y(b_1, \underline D_2) \mid X_2(b_1), A_1 = b_1, X_1 \} \mid A_1=b_1, X_1 \Big] \\
        &\equiv \bbE \Big[ \bbE \big\{ Y(b_1, \underline D_2) \mid H_2(b_1) \} \mid A_1=b_1, X_1 \Big] \\
        &\hspace{-1in}= \bbE \Big[ \sum_{b_2} \bbE \big\{ Y(b_1, b_2, \underline D_3) \mid D_2(b_1) = b_2, H_2(b_1) \} \bbP \{ D_2(b_1) = b_2 \mid H_2(b_1) \} \mid A_1 =b_1, X_1 \Big] \\
        &\hspace{-1in}= \bbE \Big[ \sum_{b_2} \bbE \big\{ Y(b_1, b_2, \underline D_3) \mid A_2(b_1) = b_2, H_2(b_1) \} Q_2(b_2 \mid b_1, \overline X_2 ) \mid A_1 =b_1, X_1 \Big] \\
        &\hspace{-1.5in}= \sum_{b_2} \bbE \Big[ \bbE \big\{ Y(b_1, b_2, \underline D_3) \mid A_2 = b_2, X_2, A_1=b_1, X_1 \}  Q_2(b_2 \mid b_1, \overline X_2 ) \mid A_1 =b_1, X_1 \Big]
    \end{align*}
    where the first line follows by iterated expectations on $X_2(b_1),A_1=b_1, X_1$, the second line by Proposition~\ref{prop:consistency}, and the third by iterated expectations on $D_2(b_2), H_2(b_1)$ and then taking the expectation over $D_2(b_2)$. The fourth follows by Lemma~\ref{lem:exch-short} inside the expectation; meanwhile, the probability $Q_2$ is identified by Proposition~\ref{prop:consistency}. The final line follows again by Proposition~\ref{prop:consistency}.  Again, the conditions of the theorem and Lemma~\ref{lem:absolute-continuity} guarantee that the expectation is well-defined, by the same argument as above for $t=1$.

    \bigskip

    Repeating this process $t-2$ more times yields  
    \[
    \bbE \left\{ Y( \overline D_T ) \right\} = \sum_{\overline b_T \in \{0,1\}^T} \int_{\overline{\mathcal{X}}_T} \bbE\{ Y(\overline b_t) \mid \overline A_T = \overline b_{T}, \overline X_T = \overline x_T \} \prod_{t=1}^{T} Q_t(b_t \mid \overline b_{t-1}, \overline x_t ) d\bbP(x_t \mid \overline b_{t-1}, \overline x_{t-1}).
    \]
    \normalsize The final result follows by the consistency assumption embedded in the NPSEM.  The IPW result follows by taking the expectation over $A_t$.

    \bigskip

    Meanwhile, the identification of $\bbE (D_t)$ follows by essentially the same argument. We'll repeat the first step here:
    \[
    \bbE(D_t) = \sum_{b_1} \bbE \{ \bbE(D_t \mid D_1 = b_1, X_1) \bbP(D_1 = b_1 \mid X_1) \} = \sum_{b_1} \bbE \{ \bbE(D_t \mid A_1 = b_1, X_1) Q_1(b_1 \mid X_1) \}
    \]
    where the first equation follows by iterated expectations, the second by Lemma~\ref{lem:exch-short} to exchange $D_1=b_1$ with $A_1 = b_1$ and by the definition of $D_1$ to yield $\bbP(D_1 = b_1 \mid X_1) = Q_1(b_1 \mid X_1)$. The result follows by repeating this process $t-1$ more times.
\end{proof}

\begin{remark}
    Suppose the flip intervention was stochastic as in our suggested stochastic amendment at the end of Section~\ref{subsec:longitudinal-flip}:
    \begin{align*}
        D_{f_t}(a_t) = \one \bigg( V_t &\leq \one(a_t = 1) f_t \{H_t(\overline D_{t-1}) \} \\
        &+ \Big[ 1 - f_t\{H_t(\overline D_{t-1}) \} \Big] \bbP \{ A_t(\overline D_{t-1} = 1 \mid H_t(\overline D_{t-1}) \} \bigg).
    \end{align*}
    Then, the sequential exchangeability step would follow by Lemma~\ref{lem:exch-short-no-strong}, which only requires standard sequential randomization in Assumption~\ref{asmp:standard-seq-exch}.  Meanwhile, identification of $\bbE\{ D_{f_t}(a_t) \}$ would only require the NPSEM assumption and possibly positivity depending on the weight function, but no exchangeability assumption.
\end{remark}

\subsection{Proposition~\ref{prop:simultaneous}}

\begin{proof}
    We have
    \begin{align*}
        \bbE \big\{ Y(\overline D_T) \big\} &= \bbE \Big[ \sum_{\overline b_T} \bbE \{ Y(\overline b_T) \mid \overline D_T = \overline b_T, \overline A_T, \Pi_T \} \bbP ( \overline D_T = \overline b_T \mid \overline A_T, \Pi_T ) \Big] \\
        &= \bbE \Big[ \sum_{\overline b_T} \bbE \{ Y(\overline b_T) \mid \overline A_T, \Pi_T \} \bbP ( \overline D_T = \overline b_T \mid \overline A_T, \Pi_T  ) \Big]
    \end{align*}
    where the first line follows by iterated expectations on $\overline A_T, \Pi_T $ and then on $\overline D_T, \overline A_T, \Pi_T$ and taking the expectation over $\overline D_T$, and the second line follows because $\overline D_T \ind Y(\overline b_T) \mid \overline A_T, \Pi_T $ because $\overline D_T$ is only random via $V$ conditional on $\overline A_T, \Pi_T$.

    \bigskip
    
    \noindent The same holds for $\bbE \big\{ Y(\overline D_T^\prime) \big\}$. Then, 
    \[
    \bbE \{ Y(\overline D_T) - Y(\overline D_T^\prime) \} = \sum_{\overline b_T} \bbE \left[ \bbE \{ Y(\overline b_T) \mid \overline A_T, \overline X_T \} \big\{ \bbP(\overline D_T = \overline b_T \mid \overline A_T, \Pi_T) - \bbP(\overline D_T^\prime = \overline b_T \mid \overline A_T, \Pi_T) \big\} \right].
    \]
    For the propensity scores, we have
    \begin{align*}
        \bbP ( \overline D_T = \overline a_T \mid \overline A_T, \Pi_T  ) &= \one(\overline A_T = \overline a_T) + \one(\overline A_T \neq \overline a_T) \Pi_T, \\
        \bbP ( \overline D_T = \overline b_T \mid \overline A_T, \Pi_T  ) &= \one(\overline A_T = \overline b_T) (1 - \Pi_T) \text{ for } \overline b_T \neq \overline a_T, \\
        \bbP ( \overline D_T^\prime = \overline a_T^\prime \mid \overline A_T, \Pi_T  ) &= \one(\overline A_T = \overline a_T^\prime) + \one(\overline A_T \neq \overline a_T^\prime) \Pi_T \text{, and} \\
        \bbP ( \overline D_T = \overline b_T \mid \overline A_T, \Pi_T  ) &= \one(\overline A_T = \overline b_T) (1 - \Pi_T) \text{ for } \overline b_T \neq \overline a_T^\prime.
    \end{align*} 
    Plugging these results into the prior display yields
    \[
    \bbE \{ Y(\overline D_T) - Y(\overline D_T^\prime) \} = \bbE \left( \Big[ \bbE \{ Y(\overline a_T) \mid \overline A_T, \Pi_T \} - \bbE \{ Y(\overline a_T) \mid \overline A_T, \Pi_T \} \Big] \Pi_T \right).
    \]
    The result follows by iterated expectations on $\overline A_T, \Pi_T$.  The same argument holds for the denominator of \eqref{eq:cross-world}.
\end{proof}

\section{Proofs for Section~\ref{sec:efficiency}} \label{app:est}

\begin{remark}
    Throughout this section, we assume all estimated nuisance functions are fixed.
\end{remark}

\subsection{Helper lemmas of efficient influence function of $Q_t$}

We begin with several general helper lemmas.

\begin{lemma}
    Under the setup of Proposition~\ref{prop:eif}, $\varphi_m(Z)$ and $\varphi_Q(Z)$ are mean-zero.
\end{lemma}

\begin{proof}
    This follows by iterated expectations on $H_t$.
\end{proof}

The next two lemmas are about the efficient influence function of $\bbE \{ Q_t(b_t \mid H_t ) \}$ and its estimator as constructed in the body of the paper.  

\begin{lemma} \label{lem:phi-first-order}
    Under the setup of Proposition~\ref{prop:eif},
    \[
    \bbE \{ \phi_t(b_t; A_t, H_t) \mid H_t \} = 0.
    \]
    and
    \begin{align*}
        \bbE \{ \widehat \phi_t (b_t; A_t, H_t) \mid H_t \} = &\ \Big\{  2 \one(b_t = a_t) - 1 \Big\}  \Big\{ \bbP(A_t = a_t) - \widehat \bbP(A_t = a_t \mid H_t) \Big\} \\
        &\cdot \Big[ 1 - s_t\{ \widehat \bbP(A_t = a_t \mid H_t) \} + s_t^\prime \{ \widehat \bbP(A_t = a_t \mid H_t) \} \big\{ 1 - \widehat \bbP(A_t = a_t \mid H_t) \big\} \Big].
    \end{align*}
\end{lemma}
\begin{proof}
    These follow by iterated expectations on $H_t$.
\end{proof}

In the next lemma we omit arguments for brevity, so that $\bbP_t = \bbP(A_t = a_t \mid H_t)$ and $s_t = s_t\{ \bbP(A_t = a_t \mid H_t) \}$ and $\widehat \bbP_t$ and $\widehat s_t(a_t)$ are defined similarly. 
\begin{lemma} \label{lem:phi-second-order}
    Under the setup of Proposition~\ref{prop:eif},
    \begin{align*}
        &\bbE \{ \widehat \phi_t(b_t; A_t, H_t) + \widehat Q_t(b_t \mid H_t) - Q_t(b_t \mid H_t) \mid H_t \} = \\
        &= \Big\{ 2 \one(b_t = a_t) - 1 \Big\} \left( \left[ \frac{\widehat s_t^{\prime \prime}}{2} \left( \widehat \bbP_t - \bbP_t \right)^2 + o \big\{  ( \widehat \bbP_t -  \bbP_t )^2 \big\} \right] \left( \widehat \bbP_t - 1 \right) + \left( \widehat s - s \right) \left( \widehat \bbP_t - \bbP_t \right) \right)  
    \end{align*}
\end{lemma}

\begin{proof}
    First, note that by definition $Q_t(b_t \mid H_t) = \bbP(A_t = b_t \mid H_t) + \Big\{ 2\one(b_t = a_t)  - 1 \Big\} s_t\{ \bbP(A_t = a_t \mid H_t) \{ 1 - \bbP(A_t = a_t \mid H_t) \}$ and therefore
    \[
    \bbE \left\{ \widehat Q_t(b_t \mid H_t) - Q_t(b_t \mid H_t) \mid H_t \right\} = \Big\{ 2\one(b_t = a_t)  - 1 \Big\} \left\{ \widehat \bbP_t - \bbP_t + \widehat s_t \cdot ( 1- \widehat \bbP_t) - s_t \cdot ( 1 - \bbP_t) \right\},
    \]
    where we omit arguments on the right-hand side.  Therefore, by iterated expectations and rearranging, we have
    \begin{align*}
        &\bbE \{ \widehat \phi_t(b_t; A_t, H_t) + \widehat Q_t(b_t \mid H_t) - Q_t(b_t \mid H_t) \mid H_t \} \\
        &= \Big\{2 \one(b_t = a_t) - 1 \Big\} \left[ (\bbP_t - \widehat \bbP_t) \left\{ 1 - \widehat s_t + \widehat s_t^\prime \cdot ( 1- \widehat \bbP_t) \right\} + \widehat s_t \cdot (1- \widehat \bbP_t) - s_t \cdot (1-\bbP_t) + \widehat \bbP_t - \bbP_t \right] \\
        &= \Big\{2 \one(b_t = a_t) - 1 \Big\} \left[ \widehat s_t^\prime \cdot (\bbP_t - \widehat \bbP_t) + \widehat s_t - s_t  - \widehat \bbP_t \left\{ \widehat s_t^\prime \cdot (\bbP_t - \widehat \bbP_t) + \widehat s_t - s_t  \right\} - s_t \cdot \widehat \bbP_t - \widehat s_t \cdot (\bbP_t - \widehat \bbP_t) + s_t \cdot \bbP_t \right] \\
        &= \Big\{2 \one(b_t = a_t) - 1 \Big\} \left[ (1 - \widehat \bbP_t) \left\{ \widehat s_t^\prime \cdot (\bbP_t - \widehat \bbP_t) + \widehat s_t - s_t \right\} + (\widehat s_t - s_t) (\widehat \bbP_t - \bbP_t)  \right]
    \end{align*}
    A second-order Taylor expansion of $s_t \{ \bbP(A_t = a_t \mid H_t) \}$ yields the result. Specifically,
    \[
    (1 - \widehat \bbP_t) \left\{ \widehat s_t^\prime \cdot (\bbP_t - \widehat \bbP_t) + \widehat s_t - s_t \right\} = (\widehat \bbP_t - 1) \left[ \frac{\widehat s_t^{\prime \prime}}{2} \left( \widehat \bbP_t - \bbP_t \right)^2 + o \left\{ (\widehat \bbP_t - \bbP_t)^2 \right\} \right]
    \]
\end{proof}

\subsection{Proposition~\ref{prop:eif} and Theorem~\ref{thm:flip-bias}}

Now, we turn to establishing Proposition~\ref{prop:eif} and Theorem~\ref{thm:flip-bias}. As discussed in the body of the paper, this can be established in two ways, by unwinding the error backwards-in-time or forwards-in-time. We start with lemmas for the backwards-in-time bound, which is similar to Lemmas 5 \& 6 in \citet{kennedy2019nonparametric}. We establish the result in full, for two reasons. First, for completeness. And second, our analysis yields a different bound on the bias. Then, we establish the forwards-in-time bound. This mirrors results in \citet{diaz2023nonparametric} and others, but is new because it accounts for estimating the $Q_t$.

\bigskip

In what follows, let $\widetilde m_t(A_t, H_t) = \bbE \left\{ \sum_{b_{t+1}} \widehat m_{t+1}(b_{t+1}, H_{t+1}) \widehat Q_{t+1}(b_{t+1} \mid H_{t+1}) \mid A_t, H_t \right\}$ as in the body of the paper.  In other words, $\widetilde m_t$ is the true sequential regression function at timepoint $t$ where all the future information is estimated.

\subsubsection{Backwards-in-time lemmas}

\begin{lemma} \label{lem:backwards-1}
    Under the setup of Proposition~\ref{prop:eif}, 
    \begin{align*}
        \bbE \{ \widehat \varphi_m(Z) \} &= m_0 - \widehat m_0 \\
        &+  \sum_{t=0}^{T} \bbE \left[ \left\{ \prod_{s=0}^{t} \widehat r_s(A_s \mid H_s) - \prod_{s=0}^{t} r_s(A_s \mid H_s) \right\} \left\{ \widetilde m_t(A_t, H_t) - \widehat m_t(A_t, H_t) \right\} \right] \\
        &+ \sum_{t=1}^{T} \bbE \left[ \left\{ \prod_{s=0}^{t-1} r_s(A_s \mid H_s) \right\} \sum_{b_t} \widehat m_t (b_t , H_t) \{ \widehat Q_t(b_t \mid H_t) - Q_t(b_t \mid H_t) \} \right]
    \end{align*}
\end{lemma}

\begin{proof}
    We have
    \begin{align*}
        &\bbE\{ \widehat \varphi_m(Z) \} = \bbE \left[ \sum_{t=T}^{0} \left\{ \prod_{s=0}^{t} \widehat r_s(A_s \mid H_s) \right\} \left\{ \sum_{b_{t+1}} \widehat m_{t+1}(b_{t+1}, H_{t+1}) \widehat Q_{t+1}(b_{t+1} \mid H_{t+1}) - \widehat m_t(A_t, H_t) \right\} \right] \\
        &= \bbE \left[ \sum_{t=T}^{0} \left\{ \prod_{s=0}^{t} \widehat r_s(A_s \mid H_s) \right\} \left\{ \widetilde m_t(A_t, H_t) - \widehat m_t(A_t, H_t) \right\} \right] \\
        &= \bbE \left[ \sum_{t=T}^{0} \left\{ \prod_{s=0}^{t} \widehat r_s(A_s \mid H_s) - \prod_{s=0}^{t} r_s(A_s \mid H_s) \right\} \left\{ \widetilde m_t(A_t, H_t) - \widehat m_t(A_t, H_t) \right\} \right] \\
        &+ \bbE \left[ \sum_{t=T}^{0} \left\{ \prod_{s=0}^{t} r_s(A_s \mid H_s) \right\} \left\{ \widetilde m_t(A_t, H_t) - \widehat m_t(A_t, H_t) \right\} \right]
    \end{align*}
    where the first equality follows by definition, the second by the definition of $\widetilde m_t(A_t, H_t)$ and iterated expectations on $A_t, H_t$, and the third by adding and subtracting $\prod_{s=0}^{t} r_s(A_s \mid H_s)$. On the RHS of the final equality, the first line is second-order. Focusing on the final line in the above display, notice first that the first and last summands in the overall sum can be isolated and the sum can be re-written as
    \begin{align*}
        &\bbE \left[ \sum_{t=T}^{0} \left\{ \prod_{s=0}^{t} r_s(A_s \mid H_s) \right\} \left\{ \widetilde m_t(A_t, H_t) - \widehat m_t(A_t, H_t) \right\} \right] \\
        &= \bbE \left[ \left\{ \prod_{t=0}^{T} r_s(A_s \mid H_s) \right\} \widetilde m_T(A_T, H_T) \right] \\
        &+ \sum_{t=T-1}^{0} \bbE \left[ \left\{ \prod_{s=0}^{t} r_s(A_s \mid H_s) \right\} \widetilde m_t(A_t, H_t) - \left\{ \prod_{s=0}^{t+1} r_s(A_s \mid H_s) \right\} \widehat m_{t+1}(A_{t+1}, H_{t+1}) \right] \\
        &- \widehat m_0
    \end{align*}
    The first term equals $\psi$ because $\widetilde m_T(A_T, H_T) = \bbE(Y \mid A_T, H_T)$ and the last term is $\widehat m_0$. Meanwhile, the middle term in the above display simplifies because
    \begin{align*}
        \bbE &\left[ \left\{ \prod_{s=0}^{t+1} r_s(A_s \mid H_s) \right\} \widehat m_{t+1}(A_{t+1}, H_{t+1}) \right] \\
        &= \bbE \left[ \left\{ \prod_{s=0}^{t} r_s(A_s \mid H_s) \right\} \bbE \left\{ \sum_{b_{t+1}} \widehat m_{t+1} (b_{t+1}, H_{t+1}) Q_{t+1}(b_{t+1} \mid H_{t+1}) \mid A_t, H_t \right\} \right] 
    \end{align*}
    by iterated expectations on $A_t, H_t$. Combining like terms and the definition of $\widetilde m_t$ yield
    \begin{align*}
        &\sum_{t=T-1}^{0} \bbE \left[ \left\{ \prod_{s=0}^{t} r_s(A_s \mid H_s) \right\} \widetilde m_t(A_t, H_t) - \left\{ \prod_{s=0}^{t+1} r_s(A_s \mid H_s) \right\} \widehat m_{t+1}(A_{t+1}, H_{t+1}) \right] \\
        &= \sum_{t=T-1}^{0} \bbE \Bigg[ \left\{ \prod_{s=0}^{t} r_s(A_s \mid H_s) \right\} \bbE \bigg\{ \sum_{b_{t+1}} \widehat m_{t+1} (b_{t+1} , H_{t+1}) \{ \widehat Q_{t+1}(b_{t+1} \mid H_{t+1}) \\
        &\hspace{4in}- Q_{t+1}(b_{t+1} \mid H_{t+1}) \} \mid A_t, H_t \bigg\} \Bigg] \\
        &= \sum_{t=1}^{T} \bbE \left[ \left\{ \prod_{s=0}^{t-1} r_s(A_s \mid H_s) \right\} \sum_{b_t} \widehat m_t (b_t , H_t) \{ \widehat Q_t(b_t \mid H_t) - Q_t(b_t \mid H_t) \} \right]
    \end{align*}
    where the last line follows by re-indexing the sum and iterated expectations on $A_t, H_t$. 
\end{proof}

\begin{lemma} 
    Under the setup of Proposition~\ref{prop:eif}, 
    \begin{align*}
        \bbE \{ \widehat \varphi_Q(Z) \} &= \bbE \left[ \sum_{t=T}^{1} \left\{ \prod_{s=1}^{t-1} \widehat r_s(A_s \mid H_s) \right\} \sum_{b_t} \widehat m_t(b_t, H_t) \bbE \{ \widehat \phi_t(b_t; A_t, H_t) \mid H_t \} \right] \\
        &+ \bbE \left[ \sum_{t=T}^{1} \left\{ \prod_{s=1}^{t-1} \widehat r_s(A_s \mid H_s) - \prod_{s=1}^{t-1} r_s(A_s \mid H_s) \right\} \sum_{b_t} \widehat m_t(b_t, H_t) \bbE\{ \widehat \phi_t(b_t; A_t, H_t) \mid H_t \} \right] \\
        &\hspace{-0.5in}+ \bbE \left( \sum_{t=T}^{1} \left\{ \prod_{s=1}^{t-1} r_s(A_s \mid H_s) \right\} \sum_{b_t} \widehat m_t(b_t, H_t) \Big[ \bbE\{ \widehat \phi_t(b_t; A_t, H_t) \mid H_t \} + \widehat Q_t(b_t \mid H_t) - Q_t(b_t \mid H_t) \Big] \right) \\
        &+ \bbE \left[ \sum_{t=T}^{1} \left\{ \prod_{s=1}^{t-1} r_s(A_s \mid H_s) \right\} \sum_{b_t} \widehat m_t(b_t \mid H_t) \Big\{ Q_t(b_t \mid H_t) - \widehat Q_t(b_t \mid H_t) \Big\} \right]
    \end{align*}
\end{lemma}

\begin{proof}
    We have
    \begin{align*}
        &\bbE \{ \widehat \varphi_Q(Z) \} = \bbE \left[ \sum_{t=T}^{1} \left\{ \prod_{s=1}^{t-1} \widehat r_s(A_s \mid H_s) \right\} \sum_{b_t} \widehat m_t(b_t, H_t) \widehat \phi_t(b_t; A_t, H_t) \right] \\
        &= \bbE \left[ \sum_{t=T}^{1} \left\{ \prod_{s=1}^{t-1} \widehat r_s(A_s \mid H_s) \right\} \sum_{b_t} \widehat m_t(b_t, H_t) \bbE \{ \widehat \phi_t(b_t; A_t, H_t) \mid H_t \} \right] \\
        &= \bbE \left[ \sum_{t=T}^{1} \left\{ \prod_{s=1}^{t-1} \widehat r_s(A_s \mid H_s) - \prod_{s=1}^{t-1} r_s(A_s \mid H_s) \right\} \sum_{b_t} \widehat m_t(b_t, H_t) \bbE\{ \widehat \phi_t(b_t; A_t, H_t) \mid H_t \} \right] \\
        &+ \bbE \left( \sum_{t=T}^{1} \left\{ \prod_{s=1}^{t-1} r_s(A_s \mid H_s) \right\} \sum_{b_t} \widehat m_t(b_t, H_t) \Big[ \bbE\{ \widehat \phi_t(b_t; A_t, H_t) \mid H_t \} + \widehat Q_t(b_t \mid H_t) - Q_t(b_t \mid H_t) \Big] \right) \\
        &+ \bbE \left[ \sum_{t=T}^{1} \left\{ \prod_{s=1}^{t-1} r_s(A_s \mid H_s) \right\} \sum_{b_t} \widehat m_t(b_t \mid H_t) \Big\{ Q_t(b_t \mid H_t) - \widehat Q_t(b_t \mid H_t) \Big\} \right]
    \end{align*}
    where the second equality follows by adding zero several times.
\end{proof}

\begin{lemma} \label{lem:eif-second-order} 
    Under the setup of Proposition~\ref{prop:eif}, 
    \begin{align*}
        \bbE \{ \widehat \varphi(Z) \} &= m_0 - \widehat m_0 \\
        &+ \sum_{t=0}^{T} \bbE \left[ \left\{ \prod_{s=0}^{t} \widehat r_s(A_s \mid H_s) - \prod_{s=0}^{t} r_s(A_s \mid H_s) \right\} \left\{ \widetilde m_t(A_t, H_t) - \widehat m_t(A_t, H_t) \right\} \right] \\
        &+ \sum_{t=1}^{T} \bbE \left[ \left\{ \prod_{s=1}^{t-1} \widehat r_s(A_s \mid H_s) - \prod_{s=1}^{t-1} r_s(A_s \mid H_s) \right\} \sum_{b_t} \widehat m_t(b_t, H_t) \bbE\{ \widehat \phi_t(b_t; A_t, H_t) \mid H_t \} \right] \\
        &+ \sum_{t=1}^{T} \bbE \left(  \left\{ \prod_{s=1}^{t-1} r_s(A_s \mid H_s) \right\} \sum_{b_t} \widehat m_t(b_t, H_t) \Big[ \bbE\{ \widehat \phi_t(b_t; A_t, H_t) \mid H_t \} + \widehat Q_t(b_t \mid H_t) - Q_t(b_t \mid H_t) \Big] \right).
    \end{align*}
\end{lemma}

\begin{proof}
    The final lines in the display in the previous two lemmas cancel, yielding the result.
\end{proof}

\subsubsection{Forwards-in-time lemmas}

\begin{lemma} \label{lem:forward-m-error}
    Under the setup of Proposition~\ref{prop:eif},
    \begin{align*}
        &\bbE \{ \widehat \varphi_m(Z) \} = m_0 - \widehat m_0 \\
        &+ \sum_{t=1}^{T} \bbE \left( \left\{ \prod_{s=0}^{t-1} \widehat r_s(A_s \mid H_s) \right\} \left[ \sum_{b_t} \left\{ \widehat m_t(b_t, H_t) - m_t (b_t, H_t) \right\} \widehat r_t(b_t \mid H_t) \left\{  \bbP(b_t \mid H_t) - \widehat \bbP(b_t \mid H_t) \right\} \right] \right) \\
        &+ \sum_{t=1}^{T}  \bbE \left( \left\{ \prod_{s=0}^{t-1} \widehat r_s(A_s \mid H_s) \right\} \left[ \sum_{b_t} m_t (b_t, H_t) \left\{  \widehat Q_t (b_t \mid H_t) -  Q_t(b_t \mid H_t) \right\} \right] \right)
    \end{align*}
\end{lemma}

\begin{proof}
    We have
    \begin{align*}
        &\bbE\{ \widehat \varphi_m(Z) \} = \bbE \left[ \sum_{t=T}^{0} \left\{ \prod_{s=0}^{t} \widehat r_s(A_s \mid H_s) \right\} \left\{ \sum_{b_{t+1}} \widehat m_{t+1}(b_{t+1}, H_{t+1}) \widehat Q_{t+1}(b_{t+1} \mid H_{t+1}) - \widehat m_t(A_t, H_t) \right\} \right] \\
        &= \bbE \left[ \left\{ \prod_{s=1}^{T} \widehat r_s(A_s \mid H_s) \right\} Y \right] - \widehat m_0 \\
        &\hspace{0.1in}+ \sum_{t=1}^{T} \bbE \left[ \left\{ \prod_{s=0}^{t-1} \widehat r_s(A_s \mid H_s) \right\} \left\{ \sum_{b_t} \widehat m_t(b_t, H_t) \widehat Q_t(b_t \mid H_t) - \widehat m_t(A_t, H_t) \widehat r_t(A_t \mid H_t) \right\} \right] \\
        &= \bbE \left[ \left\{ \prod_{s=1}^{T} \widehat r_s(A_s \mid H_s) \right\} Y \right] - \widehat m_0 \\
        &\hspace{0.1in}+ \sum_{t=1}^{T} \bbE \left( \left\{ \prod_{s=0}^{t-1} \widehat r_s(A_s \mid H_s) \right\} \left[ \sum_{b_t} \left\{ \widehat m_t(b_t, H_t) - m_t (b_t, H_t) \right\} \widehat r_t(b_t \mid H_t) \left\{  \bbP(b_t \mid H_t) - \widehat \bbP(b_t \mid H_t) \right\} \right] \right) \\
        &\hspace{0.1in}+ \sum_{t=1}^{T} \bbE \left[ \left\{ \prod_{s=0}^{t-1} \widehat r_s(A_s \mid H_s) \right\} \left\{ \sum_{b_t} m_t(b_t, H_t) \widehat Q_t(b_t \mid H_t) -  m_t(A_t, H_t) \widehat r_t(A_t \mid H_t) \right\} \right]
    \end{align*}
    where the first line follows by definition, the second by rearranging the sum, and the third by adding and subtracting $m_t$. The second line in the final expression follows by taking iterated expectations on $H_t$ and gather terms. We do not manipulate the second term in the final expression above any further because it appears in the final result. Focusing on the final term, we have
    \begin{align*}
        &\sum_{t=1}^{T} \bbE \left[ \left\{ \prod_{s=0}^{t-1} \widehat r_s(A_s \mid H_s) \right\} \left\{ \sum_{b_t} m_t(b_t, H_t) \widehat Q_t(b_t \mid H_t) -  m_t(A_t, H_t) \widehat r_t(A_t \mid H_t) \right\} \right] \\
        &= \sum_{t=1}^{T} \bbE \left( \left\{ \prod_{s=0}^{t-1} \widehat r_s(A_s \mid H_s) \right\} \left[ \sum_{b_t} m_t(b_t, H_t) \left\{ \widehat Q_t(b_t \mid H_t) - Q_t (b_t \mid H_t) \right\}  \right] \right) \\
        &\hspace{0.1in}+ \sum_{t=1}^{T} \bbE \left( \left\{ \prod_{s=0}^{t-1} \widehat r_s(A_s \mid H_s) \right\} \left[ \left\{ \sum_{b_t} m_t(b_t, H_t) Q_t (b_t \mid H_t) \right\} - m_t (A_t, H_t) \widehat r_t(A_t \mid H_t) \right] \right) \\
        &= \sum_{t=1}^{T} \bbE \left( \left\{ \prod_{s=0}^{t-1} \widehat r_s(A_s \mid H_s) \right\} \left[ \sum_{b_t} m_t(b_t, H_t) \left\{ \widehat Q_t(b_t \mid H_t) - Q_t (b_t \mid H_t) \right\}  \right] \right) \\
        &\hspace{0.1in}+ \sum_{t=1}^{T} \bbE \left( \left\{ \prod_{s=0}^{t-1} \widehat r_s(A_s \mid H_s) \right\} \left[ m_t(A_t, H_t) \left\{ r_t(A_t \mid H_t) - \widehat r_t(A_t \mid H_t) \right\} \right] \right) 
    \end{align*}
    where the first equality follows by adding and subtracting $Q_t$ and the second equality by iterated expectation on $H_t$ and gathering terms, and the final line by adding and subtracting $m_t$. The first term in the final display appears in the result, so we manipulate them no further.

    \bigskip

    Combining the left over terms, we have 
    \begin{align*}
        &\sum_{t=1}^{T} \bbE \left( \left\{ \prod_{s=0}^{t-1} \widehat r_s(A_s \mid H_s) \right\} \left[ m_t(A_t, H_t) \left\{ r_t(A_t \mid H_t) - \widehat r_t(A_t \mid H_t) \right\} \right] \right) + \bbE \left[ \left\{ \prod_{s=1}^{T} \widehat r_s(A_s \mid H_s) \right\} Y \right] - \widehat m_0 \\
        &= m_0 - \sum_{t=1}^{T} \bbE \left[ \left\{ \prod_{s=1}^{t-1} \widehat r_s(A_s \mid H_s) \right\} \left[ \left\{ m_t(A_t, H_t) - m_{t+1}(A_{t+1}, H_{t+1}) r_{t+1}(A_{t+1} \mid H_{t+1}) \right\} \widehat r_t(A_t \mid H_t) \right] \right) - \widehat m_0 \\
        &= m_0 - \widehat m_0
    \end{align*}
    where the first equality follows by taking the first term out of the initial sum, which equals $m_0$ because it is $\bbE \{ m_1(A_1, H_1) r_1(A_1 \mid H_1) \}$, and adding $\bbE \left[ \left\{ \prod_{s=1}^{T} \widehat r_s(A_s \mid H_s) \right\} Y \right]$ into the sum and combining terms, and the second equality follows by iterated expectations on $H_t$. Combining all the algebra above yields the result.
\end{proof}

\begin{lemma} \label{lem:forward-phi-error}
    Under the setup of Proposition~\ref{prop:eif},
    \begin{align*}
        &\bbE \{ \widehat \varphi_Q(Z) \} = \\
        &+ \sum_{t=1}^{T}  \bbE \left( \left\{ \prod_{s=0}^{t-1} \widehat r_s(A_s \mid H_s) \right\} \left[ \sum_{b_t} \left\{ \widehat m_t(b_t, H_t) -  m_t (b_t, H_t) \right\} \bbE \left\{ \widehat \phi_t(b_t; A_t, H_t) \mid H_t) \right\} \right] \right)\\
        &+ \sum_{t=1}^{T}  \bbE \left( \left\{ \prod_{s=0}^{t-1} \widehat r_s(A_s \mid H_s) \right\} \left[ \sum_{b_t}  m_t (b_t, H_t) \left\{  \widehat Q_t (b_t \mid H_t) -  Q_t(b_t \mid H_t) + \bbE \left\{ \widehat \phi_t(b_t; A_t, H_t) \mid H_t) \right\} \right\} \right] \right) \\
        &+ \sum_{t=1}^{T}  \bbE \left( \left\{ \prod_{s=0}^{t-1} \widehat r_s(A_s \mid H_s) \right\} \left[ \sum_{b_t}  m_t (b_t, H_t) \left\{  Q_t (b_t \mid H_t) - \widehat Q_t(b_t \mid H_t) \right\} \right] \right)
    \end{align*}
\end{lemma}

\begin{proof}
    We have
    \begin{align*}
        &\bbE \{ \widehat \varphi_Q(Z) \} = \bbE \left[ \sum_{t=T}^{1} \left\{ \prod_{s=1}^{t-1} \widehat r_s(A_s \mid H_s) \right\} \sum_{b_t} \widehat m_t(b_t, H_t) \widehat \phi_t(b_t; A_t, H_t) \right] \\
        &= \sum_{t=1}^{T}  \bbE \left[\left\{ \prod_{s=1}^{t-1} \widehat r_s(A_s \mid H_s) \right\} \sum_{b_t} \widehat m_t(b_t, H_t) \bbE \{ \widehat \phi_t(b_t; A_t, H_t) \mid H_t \} \right] \\
        &= \sum_{t=1}^{T}  \bbE \left[\left\{ \prod_{s=1}^{t-1} \widehat r_s(A_s \mid H_s) \right\} \sum_{b_t} \left\{ \widehat m_t(b_t, H_t) - m_t(b_t, H_t) \right\} \bbE \{ \widehat \phi_t(b_t; A_t, H_t) \mid H_t \} \right] \\
        &\hspace{0.1in}+ \sum_{t=1}^{T}  \bbE \left( \left\{ \prod_{s=1}^{t-1} \widehat r_s(A_s \mid H_s) \right\} \sum_{b_t}  m_t(b_t, H_t) \left[ \bbE \{ \widehat \phi_t(b_t; A_t, H_t) \mid H_t \} + \widehat Q_t(b_t \mid H_t) - Q_t(b_t \mid H_t) \right] \right) \\
        &\hspace{0.1in}+\sum_{t=1}^{T}  \bbE \left[ \left\{ \prod_{s=1}^{t-1} \widehat r_s(A_s \mid H_s) \right\} \sum_{b_t}  m_t(b_t, H_t) \left\{ Q_t(b_t \mid H_t) - \widehat Q_t(b_t \mid H_t) \right\} \right]
    \end{align*}
    where the second equality follows by adding zero several times.
\end{proof}

\begin{lemma} \label{lem:forward-combine-error}
    Under the setup of Proposition~\ref{prop:eif},
    \begin{align*}
        &\bbE \{ \widehat \varphi_m(Z) + \widehat \varphi_Q(Z) \} = m_0 - \widehat m_0 \\
        &+ \sum_{t=1}^{T} \bbE \left( \left\{ \prod_{s=0}^{t-1} \widehat r_s(A_s \mid H_s) \right\} \left[ \sum_{b_t} \left\{ \widehat m_t(b_t, H_t) - m_t (b_t, H_t) \right\} \widehat r_t(b_t \mid H_t) \left\{  \bbP(b_t \mid H_t) - \widehat \bbP(b_t \mid H_t) \right\} \right] \right) \\
        &+ \sum_{t=1}^{T}  \bbE \left( \left\{ \prod_{s=0}^{t-1} \widehat r_s(A_s \mid H_s) \right\} \left[ \sum_{b_t} \left\{ \widehat m_t(b_t, H_t) -  m_t (b_t, H_t) \right\} \bbE \left\{ \widehat \phi_t(b_t; A_t, H_t) \mid H_t) \right\} \right] \right)\\
        &+ \sum_{t=1}^{T}  \bbE \left( \left\{ \prod_{s=0}^{t-1} \widehat r_s(A_s \mid H_s) \right\} \left[ \sum_{b_t}  m_t (b_t, H_t) \left\{  \widehat Q_t (b_t \mid H_t) -  Q_t(b_t \mid H_t) + \bbE \left\{ \widehat \phi_t(b_t; A_t, H_t) \mid H_t) \right\} \right\} \right] \right)
    \end{align*}
\end{lemma}

\begin{proof}
    The final lines in the display in the previous two lemmas cancel, yielding the result.
\end{proof}

\subsubsection*{Proof of Proposition~\ref{prop:eif}}

\begin{proof}
   Lemmas~\ref{lem:eif-second-order}, \ref{lem:phi-first-order}, and \ref{lem:phi-second-order} imply that $\bbE \{ \widehat \varphi(Z) - \varphi(Z) \}$ is a second-order product of errors in nuisance functions. By the same argument, the functional satisfies a von Mises expansion with second-order remainder term. The result follows by \citet[Lemma 2]{kennedy2023semiparametric}, combined with the fact that $\bbV \{ \varphi(Z) \}$ is bounded because the outcome has bounded variance and $\frac{Q_t(A_t \mid H_t)}{\bbP(A_t \mid H_t)}$ is bounded by assumption.
\end{proof}

\subsection{Theorem~\ref{thm:flip-bias}}

\begin{proof} 
    The minimum in the result will follow by taking the minimum of the two bounds we prove below.

    \bigskip
    
    \noindent {\large \textbf{Backwards-in-time:}} \\
    The estimator is $\bbP_n \{ \widehat m_0 +  \widehat \varphi(Z) \}$. Because we have iid observations, the bias then satisfies
    \[
    \bbE \left( \widehat \psi - \psi \right) = \widehat m_0 + \bbE \{ \widehat \varphi(Z) \} - \psi \equiv  \bbE \{ \widehat \varphi(Z) \} + \widehat m_0 - m_0.
    \]
    Then, by Lemma~\ref{lem:eif-second-order}, 
    \begin{align*}
        \bbE \left( \widehat \psi - \psi \right) &= \sum_{t=0}^{T} \bbE \left[ \left\{ \prod_{s=0}^{t} \widehat r_s(A_s \mid H_s) - \prod_{s=0}^{t} r_s(A_s \mid H_s) \right\} \left\{ \widetilde m_t(A_t, H_t) - \widehat m_t(A_t, H_t) \right\} \right] \\
        &+ \sum_{t=1}^{T} \bbE \left[ \left\{ \prod_{s=1}^{t-1} \widehat r_s(A_s \mid H_s) - \prod_{s=1}^{t-1} r_s(A_s \mid H_s) \right\} \sum_{b_t} \widehat m_t(b_t, H_t) \bbE\{ \widehat \phi_t(b_t; A_t, H_t) \mid H_t \} \right] \\
        &\hspace{-0.5in}+ \sum_{t=1}^{T} \bbE \left(  \left\{ \prod_{s=1}^{t-1} r_s(A_s \mid H_s) \right\} \sum_{b_t} \widehat m_t(b_t, H_t) \Big[ \bbE\{ \widehat \phi_t(b_t; A_t, H_t) \mid H_t \} + \widehat Q_t(b_t \mid H_t) - Q_t(b_t \mid H_t) \Big] \right)
    \end{align*}
    Note that $r_t$ is bounded by the construction of $s_t$, while $\widehat m_t$ is bounded by assumption. Then, by H\"{o}lder's inequality, Lemmas~\ref{lem:phi-first-order} and \ref{lem:phi-second-order}, the triangle inequality, and Cauchy-Schwarz:
    \begin{align*}
        \left| \bbE \left( \widehat \psi - \psi \right) \right| &\lesssim \sum_{t=0}^T \sum_{s=1}^t \| \widehat r_s - r_s \| \| \widetilde m_t - \widehat m_t \| \\
        &+ \sum_{t=1}^T \sum_{s=1}^{t-1} \| \widehat r_s - r_s \| \Big( \| \widehat \bbP(a_t) - \bbP(a_t) \| + \| \widehat \bbP(b_t) - \bbP(b_t) \| \Big) \\
        &+ \sum_{t=1}^{T} \Big( \| \widehat \bbP(a_t) - \bbP(a_t) \|^2 + \| \widehat \bbP(b_t) - \bbP(b_t) \|^2 \Big).
    \end{align*}
    We can streamline this decomposition further, as in the statement of the result. First, note that $\widehat r_0 = r_0 = 1$ by definition. Second, with binary treatment, $\sum_{a_t \in \{0,1\}} \| \widehat \bbP(a_t) - \bbP(a_t) \| \lesssim \| \widehat \pi_t - \pi_t \|$ where $\pi_t(H_t) \equiv \bbP(A_t = 1 \mid H_t)$. Third, $\| \widehat r_s - r_s \| \lesssim \| \widehat \pi_t - \pi_t \|$ by Taylor expansion. Then, the final line above simplifies to
    \begin{align*}
        &\sum_{t=1}^{T} \sum_{s=1}^{t} \| \widehat \pi_s - \pi_s \| \| \widehat m_t - \widetilde m_t \| + \sum_{t=1}^{T} \sum_{s=1}^{t-1} \| \widehat \pi_s - \pi_s \| \| \widehat \pi_t - \pi_t \| + \sum_{t=1}^{T} \| \widehat \pi_t - \pi_t \|^2 \\
        &= \sum_{t=1}^{T} \sum_{s=1}^{t} \| \widehat \pi_s - \pi_s \| \Big( \| \widehat m_t - \widetilde m_t \| + \| \widehat \pi_t - \pi_t \| \Big).
    \end{align*}

    \bigskip
    \noindent {\large \textbf{Forwards-in-time:}}  \\
    By the same argument above and Lemma~\ref{lem:forward-combine-error},
    \begin{align*}
        &\bbE \left( \widehat \psi - \psi \right) = \\
        &\sum_{t=1}^{T} \bbE \left( \left\{ \prod_{s=0}^{t-1} \widehat r_s(A_s \mid H_s) \right\} \left[ \sum_{b_t} \left\{ \widehat m_t(b_t, H_t) - m_t (b_t, H_t) \right\} \widehat r_t(b_t \mid H_t) \left\{  \bbP(b_t \mid H_t) - \widehat \bbP(b_t \mid H_t) \right\} \right] \right) \\
        &+ \sum_{t=1}^{T}  \bbE \left( \left\{ \prod_{s=0}^{t-1} \widehat r_s(A_s \mid H_s) \right\} \left[ \sum_{b_t} \left\{ \widehat m_t(b_t, H_t) -  m_t (b_t, H_t) \right\} \bbE \left\{ \widehat \phi_t(b_t; A_t, H_t) \mid H_t) \right\} \right] \right)\\
        &+ \sum_{t=1}^{T}  \bbE \left( \left\{ \prod_{s=0}^{t-1} \widehat r_s(A_s \mid H_s) \right\} \left[ \sum_{b_t}  m_t (b_t, H_t) \left\{  \widehat Q_t (b_t \mid H_t) -  Q_t(b_t \mid H_t) + \bbE \left\{ \widehat \phi_t(b_t; A_t, H_t) \mid H_t) \right\} \right\} \right] \right)
    \end{align*}
    Note that $\widehat r_t$ is bounded by the construction of $s_t$, while $m_t$ is bounded by assumption. Then, by H\"{o}lder's inequality, Lemmas~\ref{lem:phi-first-order} and \ref{lem:phi-second-order}, the triangle inequality, and Cauchy-Schwarz:
    \begin{align*}
        \left| \bbE \left( \widehat \psi - \psi \right) \right| &\lesssim \sum_{t=1}^{T} \| \widehat m_t - m_t \| \| \widehat \bbP(b_t) -  \bbP(b_t) \| \\
        &+ \sum_{t=1}^T \| \widehat m_t - m_t \| \Big( \| \widehat \bbP(a_t) - \bbP(a_t) \| + \| \widehat \bbP(b_t) - \bbP(b_t) \| \Big) \\
        &+ \sum_{t=1}^{T} \Big( \| \widehat \bbP(a_t) - \bbP(a_t) \|^2 + \| \widehat \bbP(b_t) - \bbP(b_t) \|^2 \Big).
    \end{align*}
    We can streamline this decomposition further, as in the statement of the result. First, with binary treatment, $\sum_{a_t \in \{0,1\}} \| \widehat \bbP(a_t) - \bbP(a_t) \| \lesssim \| \widehat \pi_t - \pi_t \|$ where $\pi_t(H_t) \equiv \bbP(A_t = 1 \mid H_t)$. Second, $\| \widehat r_t - r_t \| \lesssim \| \widehat \pi_t - \pi_t \|$ by Taylor expansion. Then, the final line above simplifies to
    \begin{align*}
        \left| \bbE \left( \widehat \psi - \psi \right) \right| &\lesssim \sum_{t=1}^{T} \Big( \| \widehat m_t - m_t \| + \| \widehat \pi_t - \pi_t \| \Big) \| \widehat \pi_t - \pi_t \|.
    \end{align*}
\end{proof}

\subsection{Corollary~\ref{cor:s-flip-convergence}}

\begin{proof}
    We have
    \begin{align*}
        \widehat \psi - \psi &= \widehat m_0 + \bbP_n \{ \widehat \varphi(Z) \} - m_0 \\
        &= (\bbP_n - \bbP) \{ \varphi(Z) \} + (\bbP_n - \bbP) \{ \widehat \varphi(Z) - \varphi(Z) \} + \widehat m_0 + \bbP \{ \widehat \varphi(Z) \} - m_0 \\
        &= (\bbP_n - \bbP) \{ \varphi(Z) \} + (\bbP_n - \bbP) \{ \widehat \varphi(Z) - \varphi(Z) \} + \bbE (\widehat \psi - \psi).
    \end{align*}
    where the first line follows by definition, the second by adding zero and because $\bbP \{ \varphi(Z) \} = 0$, and the third line by the definition of the estimator $\widehat \psi$. The second term is $o_\bbP(n^{-1/2})$ by Chebyshev's inequality and the assumption that $\lVert \widehat \varphi - \varphi \rVert = o_\bbP(1)$ (cf. \citet[Lemma 2]{kennedy2020sharp}). Meanwhile, the third term equals the bias term in Theorem~\ref{thm:flip-bias}.  This is $o_\bbP(n^{-1/2})$ by assumption. Therefore,
    \[
    \sqrt{\frac{n}{\bbV\{ \varphi(Z) \}}}(\widehat \psi - \psi) = \sqrt{\frac{n}{\bbV\{ \varphi(Z) \}}} (\bbP_n - \bbP)\{ \varphi(Z) \} + o_\bbP(1) \indist N(0,1)
    \]
    by the central limit theorem and because $\bbV \{ \varphi(Z) \}$ is bounded because $Y$ has bounded variance and $r_t$ is bounded.

    \bigskip

    Finally, note that $\widehat \sigma^2 \inprob \bbV \{ \varphi(Z) \}$ because $\| \widehat \varphi - \varphi \| = o_\bbP(1)$. Therefore, the result follows by Slutsky's theorem.
\end{proof}

\subsection{Lemma~\ref{lem:seq-property}}

\begin{proof}
    The first result follows by repeated applications of iterated expectations. Notice that the residual terms are mean zero by iterated expectations, leaving only the plug-in term. Then, $\bbE \left\{ \phi_{t+1}(b_{t+1}; A_{t+1}, H_{t+1}) \mid H_{t+1} \right\} =0$ by Lemma~\ref{lem:phi-first-order}.  And finally, by definition, 
    \[
    \bbE \left\{ \sum_{b_{t+1}} m_{t+1}(b_{t+1}, H_{t+1}) Q_{t+1}(b_{t+1} \mid H_{t+1}) \mid A_t, H_t \right\} = m_t(A_t, H_t). 
    \]
    The second result follows by induction. Throughout, we will omit arguments. Starting with the final residual, when $s = T$, we have
    \begin{align*}
        \bbE \left\{ \left( \prod_{k=t+1}^T \widehat r_k \right) \left( Y - \widehat m_T \right) \mid A_t, H_t \right\} &= \bbE \left\{ \left( \prod_{k=t+1}^{T-1} \widehat r_k \right) \left( m_T - \widehat m_T \right) \widehat r_T \mid A_t, H_t \right\} \\
        &= \bbE \left\{ \left( \prod_{k=t+1}^{T-1} \widehat r_k \right) \left( m_T - \widehat m_T \right) \left( \widehat r_T - r_T \right) \mid A_t, H_t \right\} \\
        &+ \bbE \left\{ \left( \prod_{k=t+1}^{T-1} \widehat r_k \right) \left( m_T - \widehat m_T \right) r_T \mid A_t, H_t \right\}
    \end{align*}
    where the first equality follows by iterated expectations on $A_T, H_T$ and the second by adding and subtracting $r_T$. The first term in the final expression appears in the result, so we manipulate it no further.  The next step is the induction step.

    \bigskip

    Consider the second term in the display above and the penultimate residual, when $s = T-1$. We have
    \begin{align*}
        &\bbE \left\{ \left( \prod_{k=t+1}^{T-1} \widehat r_k \right) \left( m_T - \widehat m_T \right) r_T \mid A_t, H_t \right\} + \bbE \left\{ \left( \prod_{k=t+1}^{T-1} \widehat r_k \right) \left( \sum_{b_T} \widehat m_T (\widehat Q_T + \widehat \phi_T) - \widehat m_{T-1} \right) \mid A_t, H_t \right\} \\
        &=\bbE \left[ \left( \prod_{k=t+1}^{T-1} \widehat r_k \right) \left\{ \sum_{b_T} \left( m_T - \widehat m_T \right) Q_T + \sum_{b_T} \widehat m_T (\widehat Q_T + \widehat \phi_T) - \widehat m_{T-1} \right\} \mid A_t, H_t \right] \\
        &= \bbE \left[ \left( \prod_{k=t+1}^{T-1} \widehat r_k \right) \left\{ \sum_{b_T} \widehat m_T \left( \widehat Q_T + \widehat \phi_T - Q_T \right) + m_{T-1} - \widehat m_{T-1} \right\} \mid A_t, H_t \right] \\
        &= \bbE \left\{ \left( \prod_{k=t+1}^{T-1} \widehat r_k \right) \left\{ \sum_{b_T} \widehat m_T \left( \widehat Q_T + \widehat \phi_T - Q_T \right) \right\} \mid A_t, H_t \right]  \\
        &+ \bbE \left\{ \left( \prod_{k=t+1}^{T-2} \widehat r_k \right)\left( m_{T-1} - \widehat m_{T-1} \right) \left( \widehat r_{T-1} - r_{T-1} \right)  \mid A_t, H_t \right\} \\
        &+ \bbE \left\{ \left( \prod_{k=t+1}^{T-2} \widehat r_k \right)\left( m_{T-1} - \widehat m_{T-1} \right) r_{T-1}  \mid A_t, H_t \right\}
    \end{align*}
    where the first equality follows by gathering terms, the second by iterated expectations and the definition of $m_{T-1}$, and the third by adding and subtracting $r_{T-1}$. The first and second lines in the final display appear in the statement of the result. The third line can be combined with the earlier residual, for $s= T-2$ using the step we just outlined. This argument can be continued all the way to $s=t+1$.

    \bigskip

    For the final step, when $s=t+1$, we will be left with
    \begin{align*}
        &\bbE \left[ \left( m_{t+1} - \widehat m_{t+1} \right) r_{t+1} + \sum_{b_{t+1}} \widehat m_{t+1} \left( \widehat Q_{t+1} + \widehat \phi_{t+1} \right) \mid A_t, H_t \right] - m_t(A_t, H_t) \\
        &= \bbE \left[ \sum_{b_{t+1}} \left( m_{t+1} - \widehat m_{t+1} \right) Q_{t+1} + \sum_{b_{t+1}} \widehat m_{t+1} \left( \widehat Q_{t+1} + \widehat \phi_{t+1} \right) \mid A_t, H_t \right] - m_t(A_t, H_t) \\
        &= \bbE \left[ \sum_{b_{t+1}} \widehat m_{t+1} \left( \widehat Q_{t+1} + \widehat \phi_{t+1} - Q_{t+1} \right) \mid A_t, H_t \right]
    \end{align*}
    where the first equality follows by iterated expectations and the second by canceling \\
    $\bbE \left\{ \sum_{b_{t+1}} m_{t+1} Q_{t+1} \mid A_t, H_t \right\} - m_t(A_t, H_t) = 0$.
\end{proof}

\subsection{Theorem~\ref{thm:seq-dr-bias}}

\begin{proof}
    We omit arguments throughout. We have
    \begin{align*}
        \bbE \left( \widehat \psi^\ast - \psi \right) &= \bbE \left( \widehat m_0^\ast - m_0 \right) \\
        &= \bbE \left( \widehat P_1^\ast(Z) - m_0 \right) \\
        &= \sum_{t=1}^{T} \bbE\left[ \left( \prod_{k=1}^{t-1} \widehat r_k \right) \Big( m_t - \widehat m_t^\ast \Big) \Big( \widehat r_t - r_t \Big) + \left( \prod_{k=1}^{t-1} \widehat r_k \right) \left\{ \sum_{b_t} \widehat m_t^\ast \left( \widehat Q_t + \widehat \phi_t - Q_t \right) \right\} \right].
    \end{align*}
    where the first line follows by definition, the second by iid observations and the definition of $\widehat m_0^\ast$, and the third by Lemma~\ref{lem:seq-property}. Note that by construction $\widehat r_k$ is bounded and by assumption $\widehat m_t^\ast$ is bounded. Therefore, for the second summand in the final display above, Lemma~\ref{lem:phi-second-order}, H\"{o}lder's inequality, and Cauchy-Schwarz yield
    \[
    \left| \sum_{t=1}^{T} \bbE \left[ \left( \prod_{k=1}^{t-1} \widehat r_k \right) \left\{ \sum_{b_t} \widehat m_t^\ast \left( \widehat Q_t + \widehat \phi_t - Q_t \right) \right\} \right] \right| \lesssim \sum_{t=1}^{T} \| \widehat \pi_t - \pi_t \|^2.
    \]
    Meanwhile, the first summand from the final line in the initial display above can be bounded in two steps. First, we have
    \begin{align*}
        \bbE\left\{ \left( \prod_{k=1}^{t-1} \widehat r_k \right) \Big( m_t - \widehat m_t^\ast \Big) \Big( \widehat r_t - r_t \Big) \right\} &= \bbE\left\{ \left( \prod_{k=1}^{t-1} \widehat r_k \right) \Big( m_t - \widetilde m_t^\ast + \widetilde m_t^\ast - \widehat m_t^\ast \Big) \Big( \widehat r_t - r_t \Big) \right\}
    \end{align*}    
    by adding and subtracting $\widetilde m_t^\ast$. H\"{o}lder's inequality, a Taylor expansion of $\widehat r_t - r_t$, and Cauchy-Schwarz yield
    \[
    \left| \bbE\left\{ \left( \prod_{k=1}^{t-1} \widehat r_k \right) \Big( \widetilde m_t^\ast - \widehat m_t^\ast \Big) \Big( \widehat r_t - r_t \Big) \right\} \right| \lesssim \| \widehat m_t^\ast - \widetilde m_t^\ast \| \| \widehat \pi_t - \pi_t \|.
    \] 
    All that remains is to bound $\left| \bbE\left\{ \left( \prod_{k=1}^{t-1} \widehat r_k \right) \Big( m_t^\ast - \widetilde m_t^\ast \Big) \Big( \widehat r_t - r_t \Big) \right\} \right|$. However, notice that $m_t^\ast - \widetilde m_t^\ast$ can itself be expressed a second-order product of errors. Specifically,
    \begin{align*}
        \widetilde m_t^\ast - m_t^\ast &= \bbE \left\{ \widehat P_{t+1}^\ast(Z) \mid A_t, H_t \right\} - m_t^\ast(A_t, H_t) \\
        &= \sum_{s=t+1}^{T} \bbE\left\{ \left( \prod_{k=t+1}^{s-1} \widehat r_k \right) \Big( m_s - \widehat m_s^\ast \Big) \Big( \widehat r_s - r_s \Big) \mid A_t, H_t \right\} \\
        &+ \bbE \left[ \left( \prod_{k=t+1}^{s-1} \widehat r_k \right) \left\{ \sum_{b_s} \widehat m_s^\ast \left( \widehat Q_s + \widehat \phi_s - Q_s \right) \right\} \mid A_t, H_t \right],
    \end{align*}
    where the second line follows by Lemma~\ref{lem:seq-property}.  Therefore, proceeding forwards in time, from $t+1$ to $T$, we can inductively apply the same argument as at the start of this proof to obtain
    \begin{align*}
        \left| \bbE\left\{ \left( \prod_{k=1}^{t-1} \widehat r_k \right) \Big( m_t^\ast - \widetilde m_t^\ast \Big) \Big( \widehat r_t - r_t \Big) \right\} \right| &\lesssim \sum_{s=t+1}^{T} \| \widehat m_s^\ast - \widetilde m_s^\ast \| \| \widehat \pi_s - \pi_s \| + \| \widehat \pi_s - \pi_s \|^2,
    \end{align*}
    where the inequality follows because $r_t, \widehat r_t$ are bounded for all $t$ by construction and $\widehat m_t^\ast$ are bounded by assumption. Therefore, this final term is bounded by a sum of products of errors at later timepoints. This sum is of the same order as terms that already appear in the sum. In other words,
    \begin{align*}
        &\left| \sum_{t=1}^{T} \bbE\left[ \left( \prod_{k=1}^{t-1} \widehat r_k \right) \Big( m_t - \widehat m_t^\ast \Big) \Big( \widehat r_t - r_t \Big) + \left( \prod_{k=1}^{t-1} \widehat r_k \right) \left\{ \sum_{b_t} \widehat m_t^\ast \left( \widehat Q_t + \widehat \phi_t - Q_t \right) \right\} \right] \right| \\
        &\lesssim \sum_{t=1}^{T} \Bigg( \| \widehat m_t^\ast - \widetilde m_t^\ast \| \| \widehat \pi_t - \pi_t \| + \| \widehat \pi_t - \pi_t \|^2 + \sum_{s=t+1}^{T} \| \widehat m_s^\ast - \widetilde m_s^\ast \| \| \widehat \pi_s - \pi_s \| + \| \widehat \pi_s - \pi_s \|^2 \Bigg).
    \end{align*}
    We can ignore the sum from $s=t+1$ to $T$, because it of the same order as terms that appears in the first two summands. The result follows.
\end{proof}

\subsection{Corollary~\ref{cor:s-flip-seq-convergence}}

\begin{proof}
    This follows by the same argument as for Corollary~\ref{cor:s-flip-convergence}.
\end{proof}

\section{Proofs for Appendix~\ref{app:trt_prob}} \label{app:additional}

\subsection{Proposition~\ref{prop:trt_prob_eif}}

\begin{proof}
    We proceed using the same proof technique as for Proposition~\ref{prop:eif}.  Specifically, we'll show that $\bbE \{ \widehat \varphi_D(Z) + \widehat m_0 - m_0 \}$ is second order. 

    \bigskip
    
    The algebra works similarly to the proof for the debiased pseudo-outcomes in Lemma~\ref{lem:seq-property}. Starting with the summand at the final timepoint, we have
    \begin{align*}
        \bbE \{ \widehat \varphi_D(Z) \} &= \bbE \left[ \left( \prod_{s=1}^{T-1} \widehat r_s(A_s \mid H_s) \right) \left\{ \widehat Q_T(1 \mid H_T) + \widehat \phi_T(1; A_T, H_T) - \widehat m_{T-1}(A_{T-1}, H_{T-1}) \right\} \right] \\
        &= \bbE \left[ \left( \prod_{s=1}^{T-1} \widehat r_s(A_s \mid H_s) \right) \bbE \left\{ \widehat Q_T(1 \mid H_T) + \widehat \phi_T(1; A_T, H_T) - Q_T(1 \mid H_T) \mid A_{T-1}, H_{T-1} \right\} \right] \\
        &+ \bbE \left[ \left( \prod_{s=1}^{T-1} \widehat r_s(A_s \mid H_s) \right) \left\{ Q_T(1 \mid H_T) - \widehat m_{T-1}(A_{T-1}, H_{T-1}) \right\} \right] \\
        &= \bbE \left[ \left( \prod_{s=1}^{T-1} \widehat r_s(A_s \mid H_s) \right) \bbE \left\{ \widehat Q_T(1 \mid H_T) + \widehat \phi_T(1; A_T, H_T) - Q_T(1 \mid H_T) \mid A_{T-1}, H_{T-1} \right\} \right] \\
        &\hspace{-0.5in}+ \bbE \left[ \left( \prod_{s=1}^{T-2} \widehat r_s(A_s \mid H_s) \right) \left( \widehat r_{T-1} - r_{T-1} \right) \left\{ m_{T-1}(A_{T-1} \mid H_{T-1}) - \widehat m_{T-1}(A_{T-1}, H_{T-1}) \right\} \right] \\
        &+ \bbE \left[ \left( \prod_{s=1}^{T-2} \widehat r_s(A_s \mid H_s) \right) r_{T-1}(A_{T-1} \mid H_{T-1}) \left\{ m_{T-1}(A_{T-1} \mid H_{T-1}) - \widehat m_{T-1}(A_{T-1}, H_{T-1}) \right\} \right].
    \end{align*}
    Notice that the first two lines in the final expression are second-order terms. Then, combining the final term with the next summand in the efficient influence function, and omitting arguments, yields
    \begin{align*}
        &\bbE \left[ \left( \prod_{s=1}^{T-2} \widehat r_s \right) \left\{ r_{T-1} \left( m_{T-1} - \widehat m_{T-1} \right) + \sum_{b_{T-1}} \widehat m_{T-1} ( \widehat Q_{T-1} + \widehat \phi_{T-1} ) - \widehat m_{T-2} \right\}  \right] \\
        &= \bbE \left[ \left( \prod_{s=1}^{T-2}  \widehat r_s \right) \left\{ \sum_{b_{T-1}} \widehat m_{T-1} ( \widehat Q_{T-1} + \widehat \phi_{T-1} - Q_{T-1}) + m_{T-2} - \widehat m_{T-2} \right\}  \right].
    \end{align*}
    Then, we can repeat this algebra $t-2$ more times, and notice that the final $m_0 - \widehat m_0$ cancels out so that \( \bbE \{ \widehat \varphi_D(Z) \} + \widehat m_0 - m_0 \) is a second-order product of errors, because $\bbE \{ \widehat Q_s(b_s \mid H_s) + \widehat \phi(b_s; A_s, H_s) - Q_s(b_s \mid H_s) \mid H_s \}$ is a second-order product of errors, by Lemma~\ref{lem:phi-second-order}.
\end{proof}

\subsection{Lemma~\ref{lem:seq-property-trtment-prob}}

\begin{proof}
    That the pseudo-outcome is unbiased follows by iterated expectations. Meanwhile, one can analyze $\bbE \{ \widehat P_{t+1}^\ast (Z) \mid A_t, H_t \}$ in the same manner as in the proof of Proposition~\ref{prop:trt_prob_eif}, above, with the exception that the product term over $\widehat r$ starts at $s=t+1$ rather than $s=1$. By the final step one has the second-order term in the statement of the result plus $m_t(A_t, H_t)$, which cancels with the $-m_t(A_t, H_t)$ on the left-hand side of \eqref{eq:pseudo-dr-2}.
\end{proof}

\end{document}